\documentclass[11pt]{article}

\usepackage{wrapfig}
\usepackage{amsmath,amsthm, amssymb}
\usepackage{booktabs,graphicx}
\usepackage{algorithm2e}[ruled]
\usepackage{xcolor}
\usepackage{comment}
\usepackage{subcaption}
\usepackage{multirow}
\usepackage{authblk}
\usepackage{fullpage}
\usepackage{hyperref}
\newcommand{\ji}[1]{\textcolor{red}{ JI: #1}}
\newcommand{\arc}[1]{\textcolor{red}{ ARC: #1}}

\def\threshresphybrid{{m + \sqrt{2\rho n \ln \frac{8n}{\delta}}}}

\newif\ifpaper
\newtheorem{thm}{Theorem}
\newtheorem{defn}{Definition}
\newtheorem{lemma}{Lemma}
\newtheorem{fact}{Fact}
\newtheorem{claim}{Claim}

\def\calA{\mathcal{A}}

\def\calG{\mathcal{G}}
\def\calH{\mathcal{H}}
\def\calI{\mathcal{I}}

\def\calM{\mathcal{M}}

\def\calP{\mathcal{P}}

\def\calR{\mathcal{P}}
\def\calS{\mathcal{S}}

\def\calX{\mathcal{X}}
\def\calY{\mathcal{Y}}

\def\E{\mathbb{E}}
\def\R{\mathbb{R}}

\def\hd{\hat{d}}

\def\bern{\textnormal{Bernoulli}}
\def\ldp{\textsc{LDP}}
\def\bottom{\perp}

\def\rr{\textit{RR}}
\def\DO{\textit{U}}
\def\DegRRCheck{\textit{RRCheck}}
\def\DegRRNaive{\textit{SimpleRR}}
\def\DegHybrid{\textit{Hybrid}}
\def\DP{DP}
\def\DegCheck{\textit{DegCheck}}
\def\RLap{$\textit{R}_{Lap}$}

\def\tO{\tilde{O}}

\newcommand{\squishlist}{
	\begin{list}{$\bullet$}
		{
			\setlength{\itemsep}{1pt}
			\setlength{\parsep}{2pt}
			\setlength{\topsep}{4pt}
			\setlength{\partopsep}{0pt}
			\setlength{\leftmargin}{1.5em}
			\setlength{\labelwidth}{1em}
			\setlength{\labelsep}{0.5em} } }
	
\newcommand{\squishend}{
\end{list}  }
\usepackage{xcolor} 

\begin{document}

\title{Robust Locally Differentially Private Graph Analysis}

\author[1]{Amrita Roy Chowdhury\thanks{The first two authors made equal contributions.}}
\affil[1]{University of Michigan, Ann Arbor, USA}
\author[2]{Jacob Imola$^*$}
\affil[2]{University of Copenhagen, Denmark}
\author[3]{Kamalika Chaudhuri}
\affil[3]{University of California, San Diego, USA}




\maketitle

\begin{abstract}
 Locally differentially private (LDP) graph analysis allows private analysis on a graph that is distributed across multiple users.  However, such computations are vulnerable to poisoning attacks where an adversary can skew the results by submitting malformed data. In this paper, we formally study the impact of poisoning attacks for graph degree estimation protocols under LDP and make three key contributions. First, we show that existing LDP protocols are highly susceptible to poisoning. To address this, we propose novel robust protocols that exploit the natural redundancy in graphs—each edge is shared between two users—to ensure accurate degree estimation even under poisoning. Our protocols are more robust when the adversary is restricted to manipulating their inputs rather than their (noisy) responses. Second, we prove matching lower bounds, establishing that our protocols achieve optimal robustness against both input and response poisoning. These bounds also demonstrate a fundamental separation between the two attack models, consistent with observations from prior work.
Third, we conduct extensive experiments on real-world graphs across a range of practically motivated attacks, showing that our protocols are effective in practice.
\end{abstract}

\section{Introduction}\label{sec:intro}
A distributed graph is defined over a set of users, where each user only knows the edges involving them—that is, each user has access to their own adjacency list. This means each user has a local view of the graph, and no single entity has knowledge of the entire graph. A real-world example of this can be found in decentralized social media platforms, such as Mastodon, Diaspora, PeerTube, where each user (account holder) represents a node, and an edge between two users indicates they are "friends" (i.e., they follow each other). In this scenario, an untrusted aggregator, such as Mastodon itself, may attempt to compute statistics for the entire graph. However, since the edges represent sensitive information (e.g., edges reveal users' personal social connections), users cannot submit their data to the aggregator directly. Instead, they add noise to their data to achieve a local differential privacy (\ldp) guarantee before sharing it with Mastodon. \ldp~has already been deployed by major commercial organizations such as Google~\cite{Rappor1} and Apple~\cite{Apple}. 

\begin{figure}
     \centering
    \includegraphics[width=0.7\columnwidth]{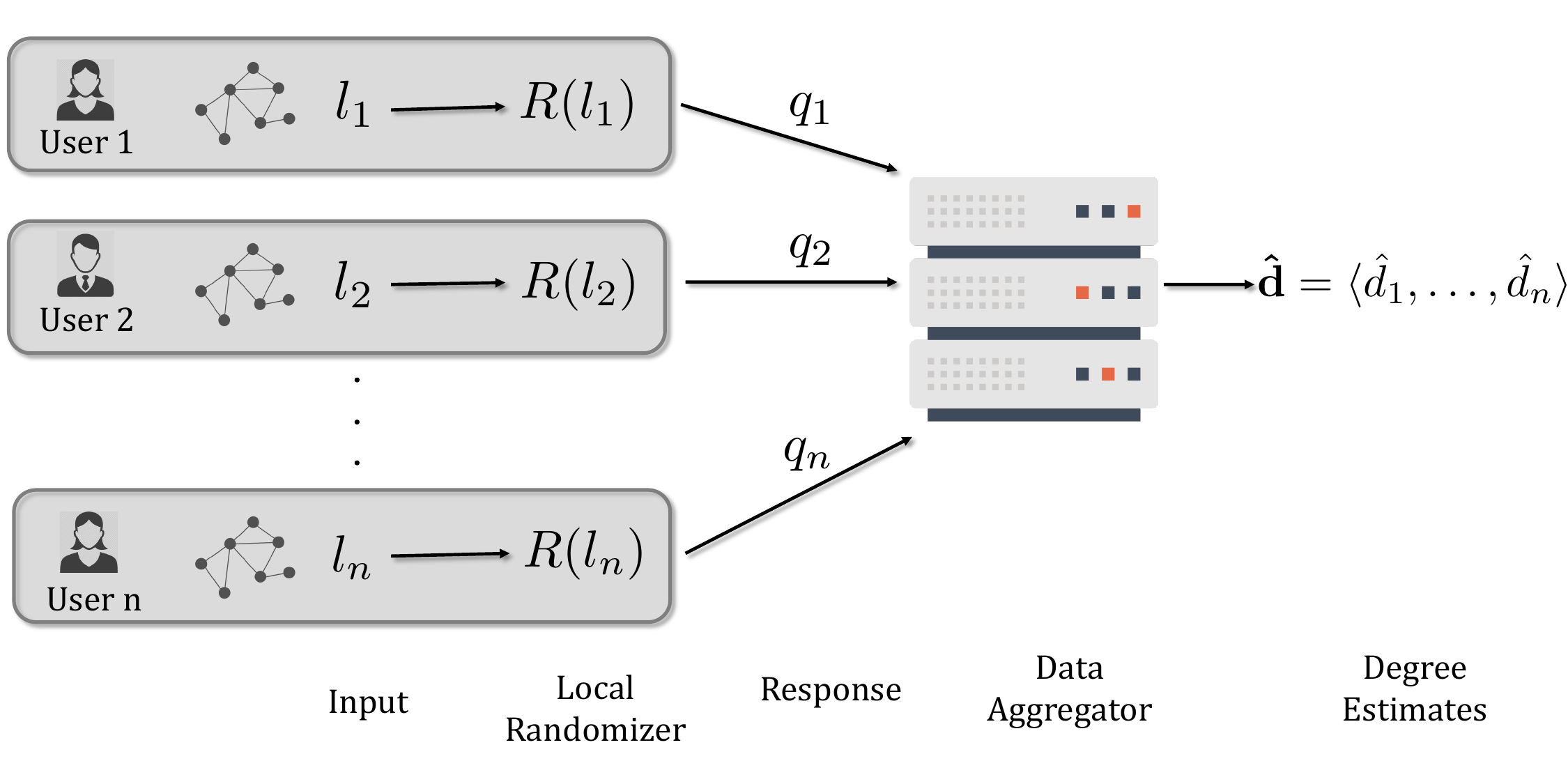}
  \caption{Analysis in the \ldp~setting}
    \label{fig:setting}   
\end{figure}
The distributed nature of \ldp, however, makes it vulnerable to poisoning attacks. For instance, it is both easy and realistic for an adversary to inject fake users into the system (e.g., by creating fake accounts on Mastodon) or compromise the accounts of real users (by hacking) to run untrusted applications on user devices. Consequently, there is no guarantee that these users will comply with the \ldp~protocols. The adversary can send carefully crafted malformed data from these non-compliant users and skew estimates, including those involving only honest users. 

Prior work, which focuses on tabular data~\cite{Cheu21,Cao21,Li22}, finds that  poisoning attacks can be carried out against \ldp{} protocols. However, the impact of poisoning under \ldp{} for graph analysis is largely unexplored. In this paper, we initiate a formal study on the impact of poisoning on \ldp{} protocols for graph statistics. We focus on the task of degree vector estimation, one of the most fundamental tasks in graph analysis~\cite{Graph11}. 

A real-world use-case for a poisoning attack is as follows -- suppose a company is interested in hiring the most influential nodes (users) of a graph for marketing its product on Mastodon and uses a node's degree as its measure of influence.  An adversary might want to promote a specific malicious node to be selected as an influencer or prevent an honest node from being selected as an influencer; concretely, suppose a single malicious user wants to be selected as influential. If the \ldp~protocol used is the Laplace mechanism, where each user directly submits their (noisy) degree to the analyst, then the malicious user can lie flagrantly and report their degree to be $n-1$, the maximum degree possible!



We  address this challenge and design degree estimation protocols that are robust to poisoning attacks. Our algorithms are based on the key observation that graph data is naturally redundant -- the information about an edge $e_{ij}$ is shared by both users $\DO_i$ and $\DO_j$. Importantly, the users do \textit{not} explicitly share this information; rather, it is implicitly shared by the structure of the graph itself. For example, in a social media graph, both users are aware of their mutual "friend" connection (i.e., the edge between them). Leveraging this observation, we propose robust protocols based on two new ideas. First, we use \textit{distributed information} -- we collect the information about each edge from \textit{both} users. The second idea is to \textit{verify} that the collected information is consistent. As long as at least one of $\DO_i$ or $\DO_j$ is honest, the analyst can check for consistency
between the two edge reports and detect malicious behavior. 

A key challenge to our mechanism design is that LDP forces all consistency checks to be probabilistic---edge inconsistencies may arise from both malicious behavior and random noise. Consequently, our protocols will flag malicious users using confidence intervals based on the expected number of inconsistent edges. We must set these intervals precisely, so that malicious users who fabricate many edges are caught, while honest users whose users are never falsely flagged as malicious. This requires us to carefully define what it means for a protcol to behave robustly.
In summary, we are the \textit{first} to study the impact of poisoning  on \ldp~degree estimation for graphs. Our main contributions are: 

\squishlist
    \item \textbf{Novel Formal Framework.} We propose a formal framework for analyzing the robustness of a protocol. 
Specifically, we measure the robustness along two dimensions, honest error and malicious error, which apply to honest and malicious users, respectively. It is important for us to delineate the error guarantees for the two types of users because the meaning of a flag is different (only malicious users should be flagged) and because it is often desirable to have a separate, better error guarantee for users who choose to follow the protocol.
\item \textbf{Lower Bounds on Poisoning Attacks.}  Based on the proposed framework, we study the impact of poisoning on degree estimation under \ldp. The attacks can be classified into two types: $(1)$ input poisoning  where the adversary does not have access to implementation of the \ldp~protocol and can only falsify their input (Fig. \ref{fig:input}), and $(2)$ response poisoning where the adversary can tamper with the \ldp~implementation and directly manipulate the (noisy) responses of the \ldp~protocol (Fig. \ref{fig:response}). 
We provide a \textit{lower bound} for input poisoning that holds for any \ldp~protocol. For response poisoning, we provide a stronger, conditional lower bound for a natural class of protocols. 

\item \textbf{Novel Robust Degree Estimation Protocols.} Leveraging the natural redundancy in graphs, we design robust degree estimation protocols under \ldp~that significantly reduce the impact of poisoning and compute degree estimates with high utility. Our robustness guarantees are \textit{attack-agnostic} -- they work for all attacks on all graphs. When $\epsilon < 0.5$
(corresponding to high privacy), our results  match our lower bounds (up to $\log(\frac{1}{\delta})$ factors) for input poisoning and response poisoning. Our results demonstrate that \ldp~makes a degree estimation protocol more vulnerable to poisoning --- response poisoning is more powerful than input poisoning.

\item  \textbf{Comprehensive Attack Evaluation.} 
We conduct a comprehensive empirical evaluation to validate our theoretical results. First, we assess the threat of poisoning attacks through \textbf{16} real-world motivated scenarios. Our findings reveal that even a small number of malicious users $(m=1\%)$ can inflict significant damage. Next, we demonstrate the robustness of our degree estimation protocols against these attacks. Our results show that our protocols effectively mitigate attacks even with a larger number of malicious users 
$(m=33\%)$ in real-world datasets.

\squishend

\ifpaper The full version of the paper is available in~\cite{Fullpaper}. \fi We have open-sourced our code in~\cite{code}.
\begin{table*}
\centering
\scalebox{0.7}{\begin{tabular}{ccc@{\extracolsep{\fill}} ccccc}
\toprule 
\multirow{2}{*}{\textbf{Protocol}} & \multicolumn{3}{c}{\textbf{Response Poisoning}} & \multicolumn{3}{c}{\textbf{Input Poisoning}} & \multirow{2}{*}{\textbf{Privacy Guarantee} }\\\cline{2-4}\cline{5-7}
& \textbf{Honest Error} & \textbf{Malicious Error} & & \textbf{Honest Error} & \textbf{Malicious Error} & & \\ 
\midrule\\
\RLap & $\tO(\frac{1}{\epsilon})$ & $n-1^*$ & Thm.~\ref{thm:response:laplace}& $\tO(\frac{1}{\epsilon})$ & $n-1 (\delta = \frac{1}{2})$ & Thm.~\ref{thm:input:laplace} & $\epsilon$-Edge \ldp\\ 
\DegRRNaive & $\tO(m + \frac{m}{\epsilon} + \frac{\sqrt{n}}{\epsilon})$ & $n-1^*$   & Thm.~\ref{thm:response:naive} & $\tO(m + \frac{\sqrt{n}}{\epsilon})$ & $n-1$ ($\delta = \frac{1}{2})$  & Thm.~\ref{thm:input:naive} & $\epsilon$-Edge \ldp\\ 
\DegRRCheck~(Ours) & $\tO(m + \frac{m}{\epsilon} + \frac{\sqrt{n}}{\epsilon})$ & $\tO(m + \frac{m}{\epsilon} + \frac{\sqrt{n}}{\epsilon})$ & Thm.~\ref{thm:response:check} & $\tO(m + \frac{\sqrt{n}}{\epsilon})$ & $\tO(m + \frac{\sqrt{n}}{\epsilon})$ & Thm.~\ref{thm:input:check} & $\epsilon$-Edge \ldp\\ 
\DegHybrid~(Ours) & $\tO(\frac{1}{\epsilon})$ & $\tO(m + \frac{m}{\epsilon} + \frac{\sqrt{n}}{\epsilon})$ & Thm.~\ref{thm:response:hybrid} & $\tO(\frac{1}{\epsilon})$ & $\tO(m + \frac{\sqrt{n}}{\epsilon})$ & Thm.~\ref{thm:input:hybrid} & $\epsilon$-Edge \ldp \\ \bottomrule
\end{tabular}}
\caption{The $\tilde{O}$ notation asymptotically holds for $\epsilon<1$, and hides factors of $\log \frac{1}{\delta}$ where $\delta$ is the probability of failure. \RLap~refers to a naive baseline approach where the users report their noisy degree using the Laplace mechanism. \DegRRNaive~refers to a naive baseline approach based on the randomized response. $^*$  indicates that there exists a worst-case attack that can \textit{always} skew the degree estimates by $n-1$. \DegRRCheck~and \DegHybrid~refer to the two robust degree estimation protocols proposed in this paper -- these protocols are \textit{optimal}, i.e., match the corresponding lower bounds (see Sec. \ref{sec:opti}). All results are attack-agnostic. 
}\label{tab:results}
\end{table*}

 \section{Preliminaries}\label{sec:background}
\textbf{Notation.}
Let $G = (V, E)$ be an undirected 
graph with $V$ and $E$ representing the set of nodes (vertices) and edges, respectively. We assume a graph with $n \in \mathbb{N}$ nodes, i.e., $V = [n]$ where $[n]$ denotes the set $\{1,2,\cdots,n\}$. Let $\calG_n$ denote the domain of all graphs with $n$ nodes. Each node $i \in V$ corresponds to a user $\DO_i$. Let $l_i \in \{0,1\}^n$ be the adjacency list for $\DO_i, i \in [n]$ where bit $l_i[j], j \in [n]$ encodes the edge $e_{ij}$ between a pair of users $\DO_i$ and $\DO_j$. Specifically, $l_i[j]=1$ if $e_{ij}\in E$ and $e_{ij}=0$ otherwise. 
Let $\textbf{d}=\langle d_1, \ldots, d_n\rangle \in \R^n$ denote the vector of degrees in $G$. $m$ denotes the number of malicious users. 
\subsection{Local Differential Privacy for Graphs}\label{sec:ldp}
Our paper focuses on the \textit{local} model of \DP~which consists of a set of individual users (\DO) and an untrusted data aggregator (analyst); each user perturbs their data using a (local) \DP~algorithm and sends it to the aggregator which uses these noisy data to estimate certain statistics of the entire dataset. 
 The most popular privacy guarantee for graphs in the local setting is known as \textit{edge \ldp}~\cite{Graph7,Graph11,qin2017generating} which protects the existence of an edge between any two users. In other words, on observing the output, an adversary cannot distinguish between two graphs that differ in a single edge. 
\begin{defn}\label{def:eldp}
    Let $R: \{0,1\}^n\mapsto \mathcal{X}$ be a randomized algorithm that takes an adajcency list $l \in \{0,1\}^n$
as input. We say $R(\cdot)$ provides $\epsilon$-edge LDP if for any two
neighboring lists $l,l' \in \{0,1\}^
n$ that differ in one bit (i.e., one
edge) and any output $s \in \mathcal{X}$,
\begin{gather} \label{eq:ledp} \mathrm{Pr}
[R(l) = s]\leq e^{\epsilon}\mathrm{Pr}[R(l') = s]
\end{gather}
\end{defn}

Randomized Response ($RR_\rho$)~\cite{RR} releases a bit $b \in \{0,1\}$ by
flipping it with probability $\rho = \frac{1}{1+e^\epsilon}$. We extend the
mechanism to inputs in $\{0,1\}^n$ by flipping each bit independently with
probability $\rho$ which satisfies
 $\epsilon$-edge \DP. 
 
The Laplace mechanism(\RLap) is a standard algorithm to achieve \DP~\cite{Dwork}. For degree estimation, each user $\DO_i$ simply reports $\tilde{d}_i=d_i+\eta, \eta \sim Lap(\frac{1}{\epsilon})$ where $Lap(b)$ represents the Laplace distribution with scale parameter $b$. This mechanism satisfies $\epsilon$-edge \DP. 

\section{Problem Overview}\label{sec:overview}
\subsection{System Setting}  

We consider \textit{ single round,  non-interactive protocols} in which each user $\DO_i, i \in [n]$ runs a local randomizer
$R_i : \{0,1\}^n \rightarrow \calX$ for some output space $\calX$ on their adjacency lists $l_i$ (Fig. \ref{fig:setting}). By \emph{non-interactive}, we mean that the local randomizers are applied independently by each user and their outputs are sent to the data aggregator in a one-shot communication. The data aggregator then collects the noisy responses and applies a post-processing function 
$D : \calX^n \rightarrow (\mathbb{N} \cup \{\bottom\})^d$ to produce final degree estimates  $\hat{\textbf{d}}=\langle \hd_1, \ldots, \hd_n\rangle$. Here, $\hd_i$ denotes the aggregator's estimate for $d_i$ for user $\DO_i$. The aggregator is also allowed to output a special symbol $\bottom$ for a user $\DO_i$ if it determines that the estimate $\hat{d}_i$ is invalid (due to suspected malicious behavior).
\subsection{Threat Model} \label{sec:threat}
  We consider a set of malicious users who may mount a data poisoning attack $\calA$. Such an attack is characterized by a subset $\calM \subseteq [n]$ of user indices corresponding to the malicious users, and a function $A$ that has access to the adjacency lists ${l_i : i \in \calM}$. The function $A$ may generate tampered inputs or response for each malicious user, depending on the specific type of poisoning attack. All malicious users are allowed to \textit{collude} together and coordinate their actions. Additionally, they have complete knowledge of the protocol, and are free to launch arbitrary, potentially adaptive attacks based on this knowledge. We do not place restrictions on how malicious users are instantiated in practice—they may be fake accounts created by the adversary, compromised legitimate accounts, or a mix of both.
\\ We refer to $\calH = [n] \setminus \calM$ as the set of indices corresponding to honest users who strictly follow the protocol.


\par Based on the specifications of the practical implementation of the \ldp, there is an important distinction between the way in which the malicious users may carry out their poisoning attacks. We outline them as follows:
\squishlist
    \item \textbf{Input Poisoning.} Here the users do not have access to the implementation of the \ldp{} randomizer. For instance, mobile applications might run proprietary code which the users do not have permission to tamper with. Consequently, an attack is only allowed to falsify input underlying input data, i.e., the function $A$ produces arbitrary false adjacency list $l_i'$ for each $i \in \calM$, and then the user's report is $q_i = R_i(l_i')$ (Fig. \ref{fig:input-atk}).
    \item \textbf{Response Poisoning.} This is a stronger threat model where a malicious user has direct control over the implementation of the \ldp{} randomizer. For instance, the user could hack into the mobile application collecting their data. Consequently, here the function $A$ produces arbitrary responses $r_i$ for $i \in \calM$, and each malicious user sends $r_{i}$ (Fig. \ref{fig:response}) to the aggregator, bypassing the application of $R_i$.
\squishend

Note that input poisoning  applies to \textit{any} protocol, private or not, because a user is free to change their input anytime. However, response poisoning attacks are unique to \ldp{} -- the distinction between an user's input and their response is a characteristic of \ldp{} which results in a separation in their efficacy (Sec. \ref{sec:input-attacks}). 
\begin{figure}
  \begin{subfigure}[b]{0.49\linewidth}
        \centering
        \includegraphics[scale=0.23]{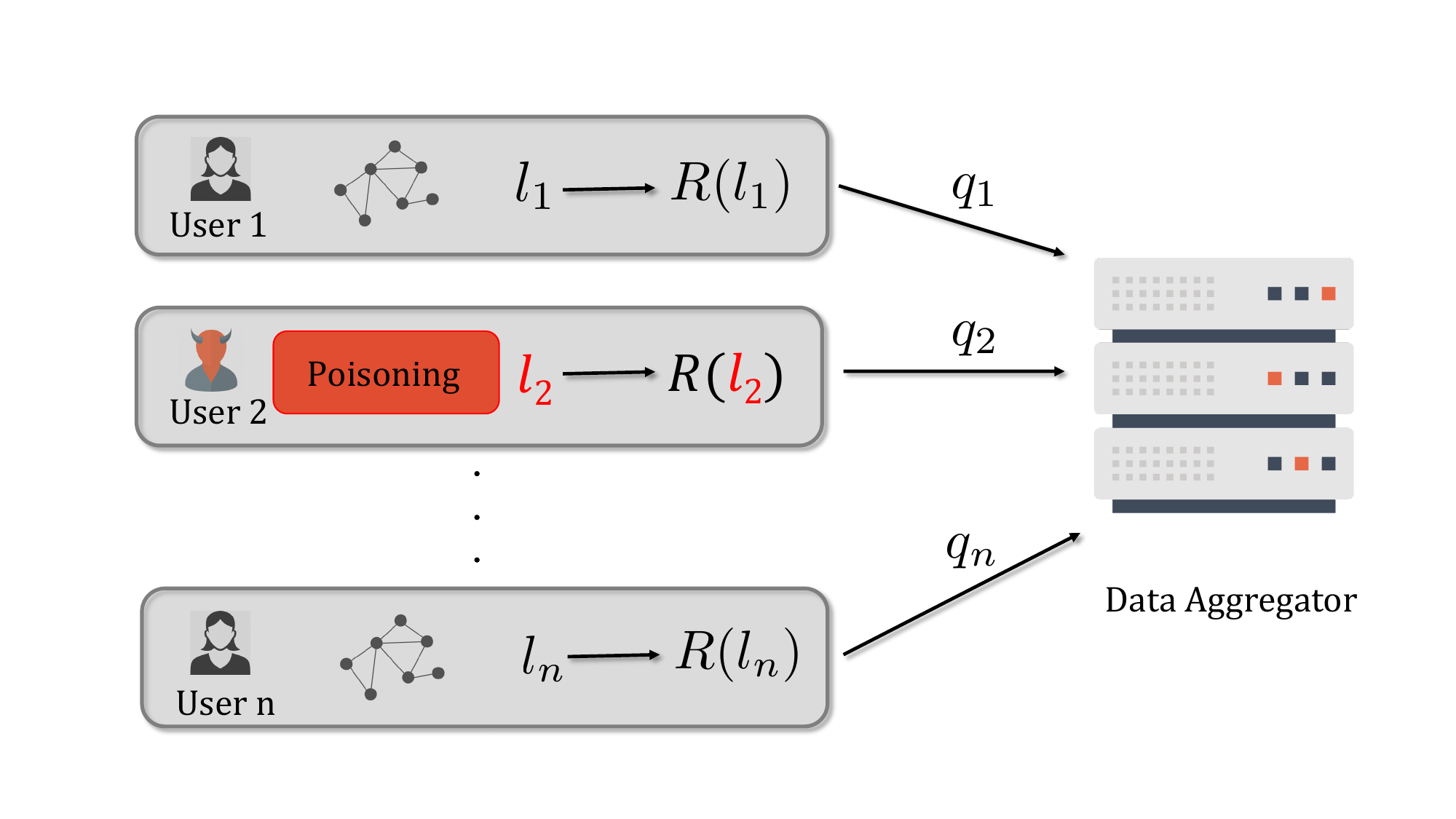}  
       \caption{Input Poisoning Attack}
    \label{fig:input-atk} 
        \end{subfigure} 
         \begin{subfigure}[b]{0.49\linewidth}
    \centering \includegraphics[scale=0.23]{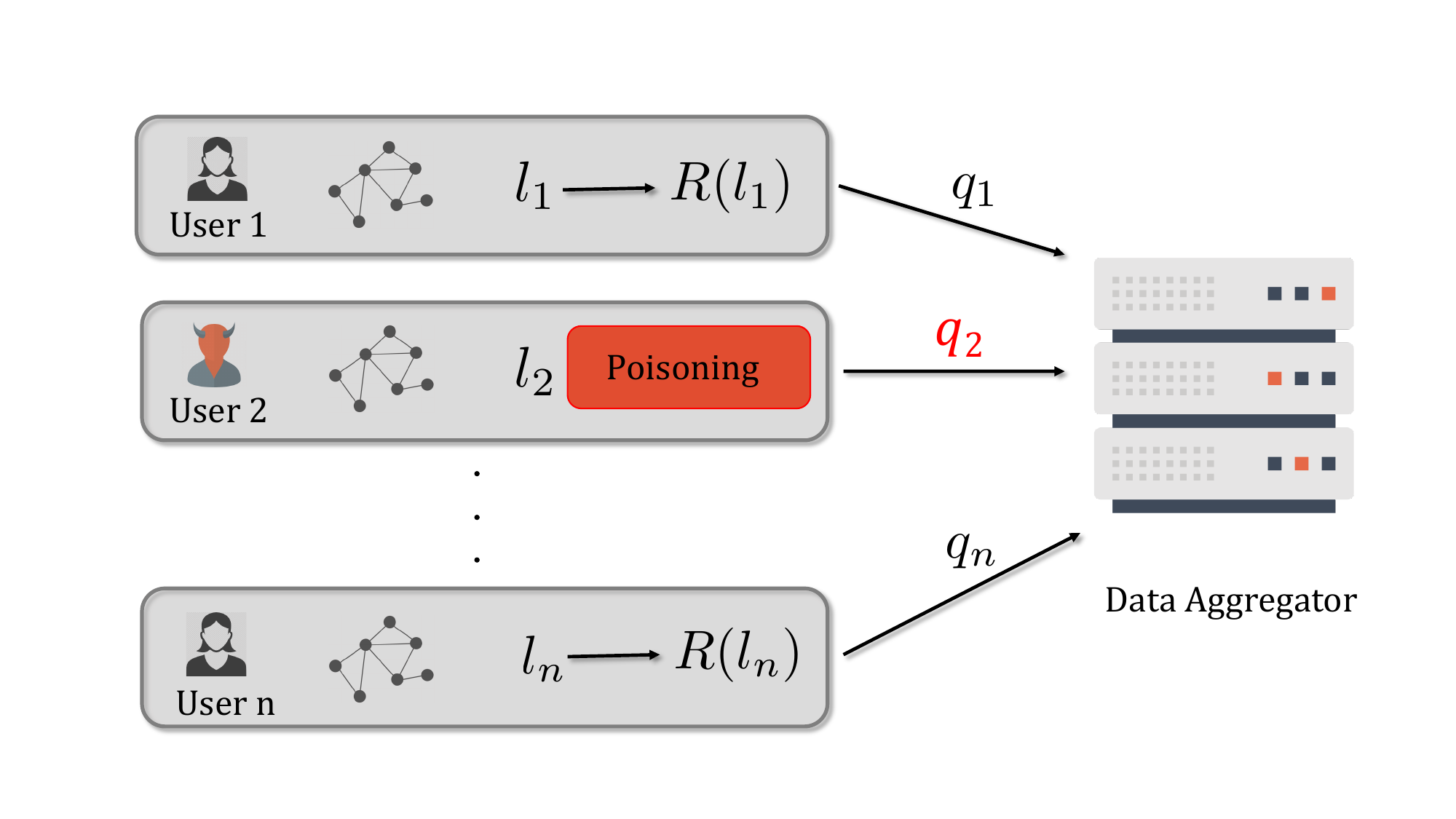} 
      \caption{Response Poisoning Attack}
    \label{fig:response}\end{subfigure}

\end{figure}

\subsection{Motivating Attacks}\label{sec:attacks} 

For the Laplace mechanism, \RLap, and randomized response mechanism, $RR_\rho$,  outlined in Sec.~\ref{sec:ldp}, we present two concrete motivating attacks. We consider the attacks in the context of the task of influential node identification, where the goal of the data aggregator is to identify certain nodes with a high measure of ``influence". This is useful for various applications where the influential nodes are selected for disseminating information/advertisement. Here, we consider degree, a simple yet often effective measure of importance. For example, nodes with highest are often selected as influencers or representatives. With the goal of modifying the degrees of different users, a malicious user may  carry out the following attacks:\\
\noindent\textbf{Degree Inflation Attack.} In this attack, a target malicious user $\DO_t$ wants to get themselves selected as an high-degree node. For \RLap, since each user sends their degree $d_i$ plus Laplace noise, the target user maximizes their degree estimate by simply sending $n-1$. 

For $RR_\rho$, the target malicious user $\DO_t$ colludes with a set of other users  (for instance, by injecting fake users) as follows. All the non-target malicious users $\DO_i, i \in \calM\setminus t$ report $1$ for the edges corresponding to $\DO_t$. Additionally, $\DO_t$ reports an all-one list.    
\\\noindent\textbf{Degree Deflation Attack.}   In this attack, the target is an honest user $\DO_t \in \calH$ who is being victimized by a set of malicious users (for instance, the adversary compromises a set of real accounts with an edge to $\DO_t$) -- $\calM$ wants to ensure that $\DO_t$ is \textit{not} selected as a high-degree node. The attack strategy is to decrease the aggregator's degree estimate $\hd_t$ for $\DO_t$. For \RLap, there is no way for the malicious party to influence $\hd_t$. However, for randomized response, the malicious users may report $0$ for each edge in $\calM$ connected to $\DO_t$, reducing their degree by this amount. 

\section{Quantifying Robustness}\label{sec:framework}

In this section, we present our formal framework for analyzing the robustness of a degree estimation protocol. Specifically, we measure the robustness along two dimensions, \textit{honest error}  and \textit{malicious error}. These assess how much the degree estimate (of either an honest or malicious user) may tampered with by an attack. We make distinguish the two types of errors because they have slightly different definitions (only malicious users should be flagged), and protocols may possibly yield smaller errors for users who follow the protocol honestly.
\\\noindent\textbf{Honest Error.} The \textit{honest error} of a protocol quantifies the manipulation of an honest user's estimator. Specifically,  malicious users can adversely affect an honest user $\DO_i \in \calH$ by $(1)$ tampering with the value of $\DO_i$'s degree estimate $\hat{d}_i$ (by introducing additional loss in accuracy), or $(2)$ attempting to mislabel $\DO_i$ as malicious (by influencing the aggregator to report $\hat{d}_i=\bot$. 


In order to account for both the above scenarios, for every honest user $\DO_i \in \calH$, we define the define the error function as the absolute difference between their true degree estimate $d_i$ and the noisy estimate $\hat{d}_i$. If $\hat{d}_i = \bot$, i.e., the user is misflagged, the error is maximized ($\infty$), since this outcome is undesirable. 
\[ 
    \mathrm{err}^{hon}(d_i, \hat{d}_i) = \begin{cases} \infty & \hat{d}_i = \bot \\ |d_i - \hat{d}_i| & \text{otherwise} \end{cases}.
\]
We seek to design protocols which are attack and graph agnostic, and thus we define honest error to consider all possible graphs and all possible attacks of size $m$, as follows:

\begin{defn}\label{def:correct}(\textbf{Honest Error}) Let 
  $\calR =\langle D, R_1, \ldots, R_n\rangle$ be a degree estimation protocol, let $m$ be a positive integer such that $m < n$, and let $\delta$ be a failure parameter. We say $\calR$ achieves $\alpha$-honest error for input (resp. response) poisoning if for all input (resp. response) poisoning attacks $\calA$ with $|\calM| \leq m$, and all graphs $G \in \calG_n$,
  \begin{gather}  
    \mathrm{Pr}\left[\textstyle{\max_{i\in \calH}} \mathrm{err}^{hon}(d_i, \hat{d}_i) \leq \alpha \right]\geq 1-\delta,\label{eq:correct1}
  \end{gather}
where the above probability is taken w.r.t  the randomness in $\calR$ and $\calA$ (both of which influence $\hat{\textbf{d}}$).
\end{defn}
\noindent\textbf{Malicious Error.} The \textit{malicious error} of a protocol quantifies the manipulations of a malicious user's estimator. 
If the protocol flags a malicious user at index $\DO_i \in \calM$ and returns $\hat{d}_i = \bottom$, this is an acceptable (and possibly desirable) outcome, and thus it incurs the minimum possible error. Otherwise, again we use the absolute difference between two degree estimates as in the following function:
\[
    \mathrm{err}^{mal}(d_i, \hat{d}_i) = \begin{cases} 0 & \hat{d}_i = \bot \\ |d_i - \hat{d}_i| & \text{otherwise} \end{cases}.
\]
Now, we are able to define malicious error in a similar way:
\begin{defn}\label{def:sound}(\textbf{Malicious Error}) Let $\calR = \langle R_1, \ldots, R_n\rangle$ be a degree estimation protocol, let $m$ be a positive integer such that $m < n$, and let $\delta$ be a failure parameter. We say $\calR$ achieves $\alpha$-malicious error for input (resp. response) poisoning if for all input (resp. response) poisoning attacks $\calA$ with $|\calM| \leq m$, and all graphs $G \in \calG_n$,
  \begin{gather}  
    \mathrm{Pr}\left[\textstyle{\max_{i\in \calM}} \mathrm{err}^{mal}(d_i, \hat{d}_i) \leq \alpha \right]\geq 1-\delta,\label{eq:sound1}
  \end{gather}
  where the above probability is taken w.r.t randomness in $\calR$ and in $\calA$.
\end{defn}

\textbf{Note.} As evident from Defns.~\ref{def:correct} and~\ref{def:sound}, our results are \emph{attack-agnostic}. In other words, they hold against \textit{all} types of attacks, including those involving arbitrary collusion strategies and even adaptive adversaries. Moreover, our results are completely general-purpose and make no assumptions about the structure of the graphs or data distribution.

\section{Robustness Lower Bounds for LDP Protocols}\label{sec:lb}

Here, we present lower bounds on the error of \ldp~protocol under poisoning. We start with input poisoning. We reduce the problem of degree vector estimation to a task of distinguishing between two scenarios, and then appeal to information-theoretic lower bounds in \ldp. In our first scenario, consider an ``honest world'' where user $\DO_n$ follows the protocol honestly, but the other malicious users manipulate their inputs to erase any edge to $\DO_n$. In the second scenario, or the ``malicious world'', user $\DO_n$ behaves maliciously and inflates his degree (using input poisoning) by $C(m + \frac{\sqrt{n}}{\epsilon})$, for a correctly chosen constant $C$, such that it matches the degree of $\DO_n$ in the honest world. The other malicious users manipulate their inputs to accordingly. The key idea is to design the input poisoning of $\DO_n$ such that his output is identical to what it would be in the honest world. Thus, the only feature left to distinguish the two worlds are the responses from $\DO_1, \ldots, \DO_{n-1}$, which are subject to information-theoretic lower bounds on LDP~\cite{duchi2013local}. This is the key to obtaining the $\frac{\sqrt{n}}{\epsilon}$ term in the error bound---if the output of $\DO_{n}$ in the malicious world were not crafted to be indistinguishable from the output in the honest world, the well-known statistical LDP bounds would not apply. Formally, our lower bound is:

\begin{thm}\label{thm:input-lb}
    For any $n, m$ such that $m < \frac{n}{10}$, $\epsilon < 0.5$, $\delta < 0.1$, there is no non-interactive, $\epsilon$-edge LDP protocol achieving both $\frac{m}{80} + \frac{\sqrt{n}}{80\epsilon}$ honest error and $\frac{m}{80} + \frac{\sqrt{n}}{80\epsilon}$ malicious error against input poisoning.
\end{thm}

The $\frac{\sqrt{n}}{\epsilon}$ term comes from the noise due to LDP, and the $m$ term comes from the malicious users modifying their inputs. Since input poisoning is a subset of response poisoning, this lower bound applies to both input and response poisoning. This theorem applies for $\epsilon < 0.5$, and typically small values of $\epsilon$ are of the most interest (corresponding to high privacy). With more careful bookkeeping of constants, the maximum value of $\epsilon$ can be increased. 

However, a response poisoning attack can do worse since malicious users are no longer constrained to follow the protocol. To carry out such an attack, these users send responses that are not valid applications of the edge-\ldp~protocol they are supposed to follow, but rather exploit the protocol's structure. In general, reasoning about arbitrary protocols that satisfy Definition~\ref{def:eldp} is challenging. Therefore, we introduce an additional structural assumption for edge-\ldp~protocols -- instead of applying the protocol globally to their entire adjacency list as in Definition~\ref{def:eldp}, users apply independent \ldp~protocols  to \textit{each} bit of their adjacency list.

\begin{defn}\label{def:efact}
    An $\epsilon$-edge LDP protocol $R$ is \emph{bitwise factorable} if it may be written as an independent (non-interactive) combination $S_1(l[1]), \ldots, S_n(l[n])$, where each $S_i:\{0,1\} \mapsto \calX_i$ satisfies \[
    \Pr[S_i(l[i]) = s_i] \leq e^\epsilon \Pr[S(l[i]') = s_i)]
    \]
    for any $l[i], l[i]' \in \{0, 1\}$, and $s_i \in \calX_i$.
\end{defn}

Randomized response, when applied to each bit, satisfies Definition~\ref{def:efact}. However, if each user releases their approximate degree using the Laplace mechanism, this is not bitwise factorable as the degree is a function of the entire adjacency list.

We provide a lower bound for response poisoning for bitwise factorable protocols as follows. 

\begin{thm}\label{thm:output-lb}
    For any $n, m$ such that $m \leq \frac{n}{10}$, $\epsilon < 0.5$, and $\delta < 0.1$, there is no non-interactive, bitwise factorable $\epsilon$-edge LDP protocol achieving $\frac{m}{80\epsilon} + \frac{\sqrt{n}}{80\epsilon}$ honest error and $\frac{m}{80\epsilon} + \frac{\sqrt{n}}{80 \epsilon}$ malicious error against response poisoning. 
\end{thm}

Intuitively, the attack uses a technical result from~\cite{kairouz2015composition} which shows that any DP mechanism with two inputs can be viewed as a post-processing of the randomized response mechanism. This allows the extension of an attack on randomized response, where malicious users always send $1$ to inflate the degree of a target, to all bitwise factorable protocols. This is key to obtaining the larger $\frac{m}{\epsilon}$ error term introduced by the malicious users, which separates input and response poisoning. We show in the next section that there are protocols which asymptotically match them for our $\epsilon < 0.5$ regime. This demonstrates that response poisoning is more powerful than input poisoning for the natural class of bitwise factorable protocols. We conjecture that the separation holds for all protocols, as well. Proofs of the above theorems are in \ifpaper the full paper~\cite{Fullpaper}\else App. \ref{app:lb-proofs} \fi.
\SetKwComment{Comment}{/* }{ */}
\RestyleAlgo{ruled} 
\SetKwInOut{Parameter}{Parameter}
\begin{algorithm}
  \caption{\DegRRNaive$: \{0,1\}^{n\times n}\mapsto \{\mathbb{N}\cup \{\bot\}\}^n$}\label{alg:degrrnaive}
  \Parameter{$\epsilon$ - Privacy parameter}
  \KwData{$\{l_1,\cdots,l_n\}$ where $l_i \in \{0,1\}^n$ is $\DO_i$'s adjacency list}
  \KwResult{$(\hat{d}_1,\cdots, \hat{d}_n)$ where $\hat{d}_i$ is $\DO_i$'s degree estimate}
  \Comment{Users}
  $\rho = \frac{1}{1+e^{\epsilon}}$\;
    \lFor{$i \in [n]$}{
      $q_i = \rr_\rho(l_i)$
    }
  \Comment{Data Aggregator}
    \For{$i \in [n]$}{
      $count_i^1 = \sum_{j < i} q_j[i] + \sum_{i < j} q_i[j]$\;
     $\hd_i = \frac{1}{1-2\rho}(count^1_i - \rho (n-1))$\;
    }
    \KwRet{$(\hd_1, \hd_2, \ldots, \hd_n)$}
\end{algorithm}

Now, we show that naive, baseline protocols  fall far from these lower bounds. 
\\\textbf{Laplace Mechanism.} The simplest mechanism for estimating degree is the Laplace mechanism, \RLap, where each user directly reports their degree estimates plus $Lap(\frac{\log(1/\delta)}{\epsilon})$ noise. Consequently, the degree estimate of an honest user \textit{cannot} be tampered with at all, and each user attains just $\frac{\log(1/\delta)}{\epsilon}$ error per user. Thus, the Laplace mechanism has $\frac{\log(1/\delta)}{\epsilon}$ honest error, and the per-user loss in accuracy is comparable with the Laplace mechanism in central DP. On the flip side, a malicious user can flagrantly lie about their estimate without detection, meaning the protocol does \emph{not} attain $\alpha$-malicious error for any $\alpha < n-1$ and $\delta < 1$. In other words, there
exists a graph and an attack against \RLap{} in
which a malicious user is \textit{guaranteed} to manipulate their true degree by $n -1$. 
\\\textbf{Randomized Response.} Now, consider the mechanism where users release their edges via randomized response. As each edge is shared by two users, edge $(i,j)$ is reported by just one of the users based on their index. We will refer to this approach as $\DegRRNaive$ (Alg. \ref{alg:degrrnaive}). The aggregator counts the total number of edges to user and then debias the estimate of the degree. Since up to $m$ of a user's edges may be reported by malicious user, we can show this protocol has $m + \frac{m+\sqrt{n \log(1/\delta)}}{\epsilon}$ honest error (with the $\sqrt{n}$ term coming from the error of randomized response). However, since a malicious user may in the worst case fabricate \textit{all} of their edges, this protocol again does \textit{not} attain $\alpha$-malicious error for any $\alpha < n-1$ and $\delta < 1$. In other words, a malicious user can always get away with the worst-case $n-1$ manipulation. Details are in App.~\ref{app:baselines}.

\section{Improving Malicious Error with Verification}\label{sec:robust-rr-checks}  

In this section, we present our proposed protocol for robust degree estimation. 
As discussed in the previous section, the naive \DegRRNaive{} protocol offers poor malicious error. To tackle
this, we propose a new protocol, \DegRRCheck, that enhances
\DegRRNaive{} with a consistency check based on the
redundancy in graph data and flags users if they fail the check. Consequently, \DegRRCheck{} improves the malicious error significantly. We observe that with higher privacy (lower $\epsilon$), the protocol is less robust. 

\subsection{\DegRRCheck{} Protocol} \label{sec:protocol:check}   
The \DegRRCheck{} protocol is described in Alg.~\ref{alg:degrrcheck} and works as follows. \DegRRCheck{} enhances the data collected by \DegRRNaive{} with verification -- for edge $e_{ij} \in E$, instead of collecting a noisy response from just one of the users
$\DO_i$ or $\DO_j$, $\DegRRCheck{}$ collects a noisy response from \emph{both} users. This creates
data redundancy which can then be checked for consistency. Specifically, the estimator counts only those edges $e_{ij}$ for which \textit{both} $\DO_i$ and $\DO_j$ are consistent and report a $1$. 
 The count of noisy edges involving user $\DO_i$ is then given by:
\[  \textstyle{count_i^{11} = \sum_{j\in [n]\setminus i} q_{i}[j] q_{j}[i].}\]
The unbiased degree estimate of $\DO_i$ is computed as follows:
\begin{equation}\label{eq:deg-est}
    \hd_i = \frac{count_i^{11} - \rho^2(n-1)}{1-2\rho}.
\end{equation}

For robustness against malicious users, \DegRRCheck{} imposes a check on the number of instances of
inconsistent reporting ($\DO_i$ and $\DO_j$ differ in their respective bits reported for their mutual
edge $e_{ij}$). For every user $\DO_i$,  the protocol has an
additional capability of returning $\bot$ whenever the  consistency check fails, indicating
that the aggregator believes that user $\DO_i$ is malicious. The intuition is that if the user $\DO_i$ is malicious and attempts to poison a lot of the edges, then there would be a large number of inconsistent reports corresponding to the edges to honest users. \DegRRCheck{}  counts the number of inconsistent reports for user $\DO_i$ as:
\[
count_i^{01} = \textstyle{\sum_{j=1}^n (1-q_{i}[j]) q_{j}[i],}
\]
i.e., the number of edges connected to user $\DO_i$\footnote{It is symmetric (and doesn't give additional information) to count edges for which user $\DO_i$ reports $1$ and $\DO_j$ reports $0$.} for
which they reported $0$ and user $\DO_j$ reported $1$. Intuitively, the check computes the expected number of inconsistent reports assuming $\DO_i$ to be honest and flags $\DO_i$ in case the reported number is outside a confidence interval. Formally, if
\begin{equation}\label{eq:deg-check} 
    |count_{i}^{01} - \rho(1-\rho)(n-1)| \leq \tau,
\end{equation}
then set $\hd_i = \bottom$, where $\tau=m + \sqrt{3 n \rho \ln \tfrac{2}{\delta}}$ is a threshold.  This check forces a malicious user to send a response with only a small number of poisoned edges (as allowed by the threshold $\tau$), thereby significantly restricting the impact of poisoning. For example, they are not able to indicate they are connected to all users in the graph, as this would produce a large number of inconsistent edges. 

Note that due to the randomization required for \ldp, some honest users might also fail the check. However,  we observe that for two honest users $\DO_i$ and $\DO_j$,  the product term
$(1-q_{i}[j]) q_{j}[i]$ follows the $ \bern(\rho(1-\rho))$ distribution, irrespective of whether the edge $e_{ij}$ exists. Consequently $count_{i}^{01}$  is tightly concentrated around its mean.  This ensures that the probability of mislabeling an honest user (by returning $\bottom$) is low. 
\setlength{\textfloatsep}{4pt}

\begin{algorithm}
  \caption{\DegRRCheck: $\{0,1\}^{n\times n}\mapsto \{\mathbb{N}\cup \{\bot\}\}^n$}\label{alg:degrrcheck}
  \Parameter{$\epsilon$-Privacy parameter; $\tau$ - Threshold for check.}

  \KwData{ $\{l_1,\cdots,l_n\}$ where $l_i \in \{0,1\}^n$ is $\DO_i$'s adjacency list}
  \Comment{Users}
  $\rho=\frac{1}{1+e^{\epsilon}}$\;
    \lFor{$i \in [n]$}{
      $q_i = \rr_\rho(l_i)$
    }
  \Comment{Data Aggregator}
    \For{$i \in [n]$}{
      $count_i^{11} = \sum_{j \in [n] \setminus i} q_{i}[j] q_{j}[i]$\;
      $count_i^{01} = \sum_{j \in [n] \setminus i} (1-q_{i}[j])q_{j}[i]$\;
      \uIf{$|count_{i}^{01} - \rho(1-\rho)(n-1)| \leq \tau$}{
        $\hd_i = \frac{1}{1-2\rho}(count_i^{11} - \rho^2 (n-1))$\;
      }\lElse{
        $\hd_i = \bottom$
      }
    }
    \KwRet{$(\hd_1, \hd_2, \ldots, \hd_n)$ where $\hat{d}_i$ is $\DO_i$'s degree estimate}
\end{algorithm}


\begin{thm}\label{thm:response:check}
The $\DegRRCheck$ protocol run with threshold \scalebox{0.9}{$\tau = m + \sqrt{2\rho n \ln \tfrac{4n}{\delta}}$} achieves
 \[2m \left(\frac{e^\epsilon+1}{e^{\epsilon}-1}\right) + 4\sqrt{n}\frac{ \sqrt{(e^\epsilon+1)\ln \frac{4n}{\delta}}}{e^\epsilon-1}\] honest and malicious error for response poisoning.
\end{thm}
This theorem is proved in \ifpaper the full paper \cite{Fullpaper} \else App.~\ref{app:b3a3}\fi. The additional verification of \DegRRCheck{} results in a clear improvement --- the malicious users can now skew their degree estimates only by a limited amount (as determined by the threshold $\tau$)  or risk getting detected,  which results in a better malicious error. Specifically, for $\epsilon$ less than a small constant (corresponding to high privacy), a malicious user can now only skew their degree estimate by at most $ \tilde{O}\big(m(1+\frac{1}{\epsilon}) + \frac{\sqrt{n}}{{\epsilon}}\big)$ for response poisoning attacks, respectively (as compared to $n-1$ in the naive case). 


The errors worsen with smaller $\epsilon$ since at lower privacy, the collected responses are noisier thereby making it harder to distinguish honest users from malicious ones. In particular, a protocol should not return $\bottom$ for honest users (i.e., mislabel them) to ensure good accuracy. Consequently, more malicious error is tolerated before a $\bottom$ is returned for a malicious user. This is evident in Eq.~\ref{eq:deg-check}--observe the threshold $\tau$ grows with smaller $\epsilon$. We expand on this price of privacy in Sec. \ref{sec:opti}. 

Interestingly for response poisoning, the degree deflation attack (Sec. \ref{sec:attacks}) represents a worst-case attack for accuracy -- the attack can skew an honest user's degree estimate by $\Omega\big(m(1+\frac{1}{\epsilon})+\frac{\sqrt{n}}{\epsilon}\big)$.  Similarly, the degree inflation attack (Sec. \ref{sec:attacks}) can skew a malicious user's degree estimate by $\Omega\big(m(1+\frac{1}{\epsilon})+\frac{\sqrt{n}}{\epsilon}\big)$) resulting in the worst-case malicious error. 

\section{Improving Honest Error with a Hybrid Protocol}\label{sec:hybrid}
The robustness guarantees for $\DegRRCheck{}$ contain a $\tilde{O}(\frac{\sqrt{n}}{\epsilon})$ term coming from the error in randomized response.
This is inherent in \textit{any} randomized response based mechanism  ~\cite{error1,error2,error3} since each of the $n$ bits of the adjacency list need to be independently randomized. Unfortunately, this dependence on $n$ has an adverse effect on the utility of the degree estimates. Typically, real-world graphs are sparse in nature with maximum degree $d_{max}\ll n$. Hence, the $\tilde{O}(\frac{\sqrt{n}}{\epsilon})$ noise term completely dominates the degree estimates resulting in poor accuracy for the honest users. 
On the other hand, \RLap{} provides a more accurate degree estimate for the honest users but has the worst-case malicious error -- a malicious user can perturb their input by $(n-1)$. In this section, we present a mitigation strategy. The key idea is to combine the two approaches and use a hybrid protocol, \DegHybrid, that achieves the best of both worlds -- honest error of \RLap, and malicious error of \DegRRCheck.

The $\DegHybrid$ protocol is outlined in Alg. \ref{alg:deghybrid} and described as follows. Each user $\DO_i$ prepares two responses -- the noisy adjacency list, $q_i$, randomized via $\textsf{RR}_{\rho}$, and a noisy degree estimate, $\tilde{d}^{lap}_i$, perturbed via \RLap, and sends them to the data aggregator. $\DO_i$ divides the privacy budget between the two responses according to some constant $c \in (0,1)$. The data aggregator first processes each list $q_i$ to employ the same consistency check on $count^{01}_i$ as that of the \DegRRCheck{} protocol (Step 9).
In case the check passes, the aggregator computes the unbiased degree estimate $\tilde{d}^{rr}_i$ from $count_i^{11}$, in the exact same way as \DegRRCheck. Note that $\tilde{d}^{rr}_i$ and $\tilde{d}^{lap}_i$ are the noisy estimates of the \textit{same} ground truth degree, $d_i$, computed via two different randomization mechanisms. To this end, the aggregator employs a second check (Step 11) to verify the consistency of the two estimates:
\begin{gather*}
|\tilde{d}_i^{rr} - \tilde{d}_i^{lap}| \leq \frac{2\tau}{1-2\rho} + \frac{1}{(1-c)\epsilon}\ln \frac{2n}{\delta} ,\end{gather*} 
where $\rho$ in this case is equal to $\frac{1}{1+e^{c\epsilon}}$.
This check accounts for the error from $\tilde{d}^{rr}_i$ (the $\frac{2\tau}{1-2\rho}$) term, and the error from $\tilde{d}^{lap}_i$ (the $\frac{1}{(1-c)\epsilon}\ln \tfrac{2n}{\delta}$ term).
Finally, the protocol returns $\bot$ if either of the checks fail.
In the event that both the checks pass, the aggregator uses $\tilde{d}_i^{lap}$ (obtained via \RLap) as the final degree estimate $\hat{d}_i$ for $\DO_i$.

\begin{algorithm}[bt]
  \caption{\DegHybrid: $\{0,1\}^{n\times n}\mapsto \{\mathbb{N}\cup \{\bot\}\}^n$}\label{alg:deghybrid}
  \Parameter{$\epsilon$-Privacy parameter; $\tau$ - Threshold for check.}
  
  \KwData{$\{l_1,\cdots,l_n\}$ where $l_i$ is $\DO_i$'s adjacency list}
  \Comment{Users}
  Select $c\in (0,1)$ and set $\rho=\frac{1}{1+e^{c\epsilon}}$\;
    \For{$i \in [n]$}{
      $q_i = \rr_\rho(l_i)$\;
      $\tilde{d}_i^{lap} = \|l_i\|_1 + Lap(\frac{1}{(1-c)\epsilon})$\;
    }
    \Comment{Data Aggregator}
    \For{$i \in [n]$}{
      $count_i^{11} = \sum_{j \in [n] \setminus i} q_{i}[j] q_{j}[i]$\;
      $count_i^{01} = \sum_{j \in [n] \setminus i} (1-q_{i}[j])q_{j}[i]$\;
      \uIf{$|count_{i}^{01} - \rho(1-\rho)(n-1)| \leq \tau$}{
        $\tilde{d}_i^{rr} = \frac{1}{1-2\rho}(count_i^{11} - \rho^2 (n-1))$\;
        \lIf{$|\tilde{d}_i^{rr} - \tilde{d}_i^{lap}| \leq \frac{2\tau}{1-2\rho} + \frac{1}{(1-c)\epsilon}\ln \tfrac{2n}{\delta} $}{
        $\hd_i = \tilde{d}_i^{lap}$
        }\lElse{
        $ \hd_i = \bottom$
        }
      }\lElse{$\hd_i = \bottom$
      }
    }
    \KwRet{$(\hd_1, \hd_2, \ldots, \hd_n)$ where $\hat{d}_i$ is $\DO_i$'s degree estimate}
\end{algorithm}

Each $\hd^{rr}_i$ estimate is computed identically to that of \DegRRCheck{}. \DegHybrid{} allows a user to send an even more accurate estimate of their degree -- to prevent malicious users from outright lying about this value, $\hd_i^{lap}$ is compared to $\hd_i^{rr}$. This allows \DegHybrid{} to enjoy the honest error bound of \RLap{} and the malicious error bound of \DegRRCheck{}. Formally,

\begin{thm}\label{thm:rrlapchecka3}
    For all $c \in (0,1)$, the \DegHybrid{} protocol run with $\tau = \threshresphybrid$ has $\frac{\ln \frac{2n}{\delta}}{(1-c)\epsilon}$ honest error and $4m (\frac{e^{c\epsilon}+1}{e^{c\epsilon}-1}) +$ $ 8\sqrt{n}\frac{ \sqrt{(e^{c\epsilon}+1)\ln \frac{8n}{\delta}}}{e^{c\epsilon}-1} + \frac{\ln \frac{2n}{\delta}}{(1-c)\epsilon}$ malicious error for response poisoning. \label{thm:response:hybrid}
\end{thm}

The proof is in \ifpaper the full paper~\cite{Fullpaper}\else App. \ref{app:thm:rrlapchecka3}\fi. We remark that \DegHybrid{} achieves the optimal accuracy of $\tilde{O}(\frac{1}{\epsilon})$ that is achievable under \ldp. This is due to the fact that the data aggregator uses $\tilde{d}^{lap}_i$ as its final degree estimate. For $\epsilon$ less than a small constant, the malicious error can be written as $\tO(m(1+\frac{1}{\epsilon}) + \frac{\sqrt{n}}{\epsilon})$ which is the same as that of \DegRRCheck{}. This is enforced by the two consistency checks. Hence, the hybrid mechanism achieves the best of both worlds.

\section{Results for Input Poisoning Attacks}\label{sec:input-attacks}  
So far we have only considered response poisoning where the malicious users are free to report arbitrary responses to the aggregator. However, to carry out such an attack in practice, a user would have to bypass the \ldp~data collection mechanism. Concretely, if a mobile application were used to collect a user's data, a malicious user would have to hack into the software and directly report their poisoned response. On the other hand, for input poisoning the malicious user needs to just lie about their input  to the application (for instance, by misreporting their list of friends) which is always possible. Hence clearly, input poisoning attacks are more easily realizable in practice. In fact, in certain cases the malicious users might be restricted to just input poisoning attacks due to the implementation of the \ldp~mechanism. For instance, mobile applications might have strict security features in place preventing unauthorized code tampering.  Another possibility is cryptographically verifying the randomizers~\cite{Kato21} to ensure that all the steps of the privacy protocol (such as, noise generation) is followed correctly.  Given its very realistic practical threat, here we study input poisoning. 

Note that input poisoning attacks are strictly weaker than response poisoning attacks. This is because the poisoned input is randomized  to satisfy \ldp~in the former which introduces noise in the final output, thereby weakening the adversary's signal. Hence intuitively, we hope to obtain better robustness against input poisoning attacks. We analyze the robustness of \DegRRNaive{}, and our proposed protocols, \DegRRCheck{} and \DegHybrid{}, under input poisoning attacks. We defer our discussion for \DegRRNaive{} to Appendix~\ref{app:input} For both the mechanisms here, we are able to a set a smaller value for $\tau$, the threshold for checking the number of inconsistent edges. This is because the number of inconsistent edges is more concentrated around its means, and hence, a tighter confidence interval with a smaller $\tau$ suffices. Thus, both the honest and malicious errors of the protocols are improved. Formally for \DegRRCheck{}, we have:

\begin{thm}\label{thm:input:check} 
The protocol $\DegRRCheck$ run with $\tau = m(1-2\rho) + \sqrt{8 \max\{\rho n, m\} \ln \frac{8n}{\delta}}$
achieves
  \scalebox{0.9}{$2m+4\sqrt{\max\{n, m(e^\epsilon+1)\}}\frac{\sqrt{2(e^\epsilon+1) \ln \frac{8n}{\delta}}}{e^\epsilon-1}$}-honest and malicious error with respect to any input poisoning attack.
\end{thm}
The proof is in \ifpaper the full paper \cite{Fullpaper}\else App.~\ref{app:b3a2}\fi.
For typical values of $\epsilon$, the honest and malicious errors can be written as $\tilde{O}(m + \frac{\sqrt{n}}{\epsilon})$ (because $\sqrt{m(e^\epsilon+1)} \leq \sqrt{n}$). Compared to Thm.~\ref{thm:response:check} for response poisoning attacks, there is an improvement of $\frac{m}{\epsilon}$ which is a direct consequence of a smaller $\tau$. 

For \DegHybrid{}, we have:
\begin{thm}\label{thm:rrlapchecka2} 
  For any $c \in (0,1)$, \DegHybrid{} run with threshold 
$
  \tau = m(1-2\rho) + \sqrt{8 \max\{\rho n, m\} \ln \frac{8n}{\delta}}
  $
  attains $\frac{1}{(1-c)\epsilon}\ln \tfrac{4n}{\delta}$ honest error and $4m+8\sqrt{\max\{n, m(e^{c\epsilon}+1)\}}\frac{\sqrt{2(e^{c\epsilon}-1) \ln \frac{8n}{\delta}}}{e^{c\epsilon}+1}$ malicious error under input poisoning.\label{thm:input:hybrid}
\end{thm}
This theorem is proved in \ifpaper the full paper \cite{Fullpaper}\else App.~\ref{app:thm:rrlapchecka2}\fi.
Written asymptotically for small $\epsilon$, the honest error of \DegHybrid{} is $\tilde{O}(\frac{1}{\epsilon})$, and its malicious error is $\tilde{O}(m + \frac{\sqrt{n}}{\epsilon})$.
Compared with Thm.~\ref{thm:response:hybrid} for response poisoning attacks, \DegHybrid{} offers similar honest error since the data aggregator uses the degree estimate collected via \RLap{} as its final estimate as before. However, the malicious error is improved by an additive factor of $O(\frac{m}{\epsilon})$, which comes from the smaller $\tau$. 


\section{Discussion}\label{sec:opti}

\subsection{Attack-Agnostic Nature of Our Protocols}\label{sec:discussion} It is important to emphasize that all our results (Thms. \ref{thm:response:check} to \ref{thm:input:hybrid}) are completely \textit{attack-agnostic} and hold for any graph. In what follows, we discuss specific implications of this property.
\\\noindent \textbf{Collusion.} As mentioned in Sec. \ref{sec:threat}, the malicious users are free to follow arbitrary collusion strategies. A direct consequence is that our robustness guarantees must hold even in the worst-case scenario, where \textit{all} $m$ malicious users are colluding with each other and can coordinate to consistently lie about their shared edges. For example, under a degree inflation attack, a targeted malicious user can reliably expect its reported degree to be inflated by at least $m-1$ when all other malicious users collude. Formally, this is captured by the $\Omega(m)$ term in all of our results. Notably, this term is unavoidable—it also appears in our lower bounds. While these attacks fall outside the scope of our theoretical robustness guarantees, practical detection remains possible. In addition to using our protocols for data collection, the aggregator may analyze the resulting graph structure to identify potential collusion patterns. In particular, the following structures are suggestive of collusion -- (1)
Star-like subgraphs, where non-central nodes exhibit abnormally low degrees (2) Disconnected cliques that are isolated from the rest of the graph. To support such detection, one can leverage techniques from the extensive literature on collusion detection in social networks~\cite{zhang2004making, shenEnhancing2016, arora2020analyzing, dutta2022blackmarket}.  We empirically evaluate their effectiveness in Sec.~\ref{sec:eval}.
\\\noindent\textbf{Bound on Malicious Users.} 
An upper bound on the number of adversaries ($m$) is a standard and widely adopted assumption in private and robust data analysis, spanning secure multi-party computation~\cite{Goldreich2001}, coding theory~\cite{codingtheory}, and Byzantine fault tolerance~\cite{Byzantine1,Byzantine2,Byzantine3}. Even prior work on poisoning attacks under \ldp~adopts this assumption\cite{Cao21,Wu21,10415225}. In practice, $m$ can be set based on empirical evidence~\cite{9833647}.  We provide additional experiments in Sec. \ref{sec:eval}.
\\\noindent\textbf{Computational and Communication Overhead.} Since we focus on non-interactive protocols for \ldp~graph analysis, each user is limited to a single round of communication. This constraint implies that arbitrary graph analysis must be performed using only the data provided in that one interaction. The only viable approach, then, is for users to report their local adjacency lists perturbed via \rr, enabling the construction of a noisy global adjacency matrix. This matrix can then be used for downstream graph analysis. Both of our proposed protocols rely on \rr~and operate by analyzing these noisy adjacency lists. As a result, they do \textit{not} incur any asymptotic increase in computational or communication overhead compared to other non-interactive \ldp~graph analysis protocols.
\subsection{Optimality of our Protocols} Our proposed protocols are optimal with respect to the lower bounds established in Sec.~\ref{sec:lb}. For $\epsilon < 0.5$, the lower bounds require that either the honest error or the malicious error is at least $\Omega(m + \frac{\sqrt{n}}{\epsilon})$ (for input poisoning) or $\Omega(\frac{m}{\epsilon} + \frac{\sqrt{n}}{\epsilon})$ (for response poisoning). Up to $\log \frac{1}{\delta}$ factors, \DegRRCheck{} achieves these error values for both honest and malicious errors, while \DegHybrid{} achieves this level of malicious error, but with a better honest error of just $\frac{1}{\epsilon}$, which is the optimal error under central DP. In this sense, \DegHybrid{} is tight with the lower bounds for both honest and malicious users, and both algorithms are tight if we consider the maximum error over both types of users.

The first separation established by our results is that is there is a price to pay for privacy when defending against poisoning attacks. Without a privacy requirement, the data aggregator can ask all users to release their adjacency list, ensuring that each honest user has at most $m$ inconsistent edges (where one user reports the edge exists, but the other does not). By flagging users with more than $m$ inconsistent edges, the data aggregator can ensure that the degrees of the remaining users are tampered with by at most $m$. This establishes a separation from input (or response) poisoning for $\epsilon < 0.5$, where it was shown in Thm.~\ref{thm:input-lb} that any protocol must have a $\Omega(m + \frac{\sqrt{n}}{\epsilon})$) term in either its honest or malicious error. The additional $\frac{\sqrt{n}}{\epsilon}$ term is the price one pays for privacy.

Second, for $\epsilon < 0.5$, the robust \DegRRCheck{} and \DegHybrid{} protocols have honest and malicious error guarantees that are tight with the lower bounds, up to $\log(\frac{1}{\delta})$ factors, for both input poisoning (Thms.~\ref{thm:input:check}, \ref{thm:input:hybrid}, and~\ref{thm:input-lb} imply that $\Theta(m + \frac{\sqrt{n}}{\epsilon})$ honest or malicious error is necessary) and response poisoning (Thms. \ref{thm:output-lb}, \ref{thm:response:check}, and \ref{thm:response:hybrid} imply that $\Theta(\frac{m}{\epsilon} + \frac{\sqrt{n}}{\epsilon})$ honest or malicious error is necessary). This shows a separation between input and response poisoning -- essentially, in input poisoning, the best the malicious users can accomplish is to alter any degree by up to $\Theta(m)$; for response poisoning, they can introduce an error of $\Theta(\frac{m}{\epsilon})$, which is asymptotically worse for small $\epsilon$. Although this separation only holds for bitwise factorable protocols (see \ref{thm:response:hybrid}), this still constitutes a large class of protocols, including randomized response and \DegRRCheck{}. 

Finally, the $\frac{m}{\epsilon}$ extra incurred error of bitwise factorable protocols from response poisoning demonstrates the increased susceptibility of LDP protocols to poisoning compared to non-private protocols. Adversaries can amplify their error by a factor of $\frac{1}{\epsilon}$ if they tamper directly with the noisy response.
\section{Extension to Other Tasks}\label{sec:extension}
In this section, we provide an extension of our robustness results to general, real-valued graph queries based on the degree vector of the form 
$F(d_1, \ldots, d_n) = \sum_{i=1}^n f(d_i)$, where $f$ is an arbitrary real-valued function. This function can be specialized to variety of useful graph queries. For example, the choice $\sum_{i=1}^n d_i^2$ computes the number of 2-star subgraphs in $G$, while the choice $\sum_{i=1}^n \mathbf{1}[d_i \geq 100]$ computes the number of nodes with degree larger than $100$. 

We state our  results under the assumption that users have degree bounded by $\Delta \leq n$ in the true graph---this assumption is often true in real-world graphs with an enforced maximum degree, such as a social network. We now define a notion of Lipschitz-ness restricted to functions of bounded degrees:
\begin{defn}
    We say the function $f(d) : \mathbb{R} \rightarrow \mathbb{R}$ is $(L, \ell, \Delta)$-restricted Lipschitz if 
    $F$ is $\ell$-Lipschitz for all degrees less than $\Delta$, meaning $|f(d) - f(d+1)| \leq \ell$ for all $d \in [0, \Delta]$.
\end{defn}

Given this definition, we can bound the quality of the estimate produced by any estimator with $\alpha_1$ honest error and $\alpha_2$-malicious error. As $\alpha_1, \alpha_2$ may be different from each other, we state our bound in terms of two restricted Lipschitz constants of $f$ to make it as fine-grained as possible:
\begin{thm}
    Suppose $F$ is a degree query function with an $f$ that is $(\ell_1, \Delta + \alpha_1)$ and $(\ell_2, \Delta + \alpha_2)$-restricted Lipschitz. Let $\tilde{\mathbf{d}}$ be a noisy degree vector produced by an algorithm with $\alpha_1$-honest error and $\alpha_2$-malicious error. Then, the estimator \scalebox{0.89}{$\tilde{F}(\tilde{\mathbf{d}}) = \sum_{i=1}^n f(\tilde{d}_i)$} satisfies 
    \scalebox{0.89}{$
        |\tilde{F}(\tilde{\mathbf{d}}) - F(\mathbf{d})| \leq \ell_1 (n-m) \alpha_1 + \ell_2 m\alpha_2 ,
    $}
    where $\mathbf{d}$ is the true degree vector.
\end{thm}
\begin{proof}
    For each of the $n-m$ honest users, we apply the Lipschitz property $\alpha_1$ times, and for each of the $m$ malicious users, we apply the Lipschitz property $\alpha_2$ times. 
\end{proof}
In particular, for estimating $2$-stars, $\DegHybrid{}$ achieves error 
\begin{align*} 
&2(n-m) (\Delta + \tfrac{\ln(n/\delta)}{\epsilon}) + 2m(\Delta + \tfrac{m}{\epsilon} + \tfrac{\sqrt{n}}{\epsilon}) \\ &
\leq 2n \Delta + 2n \tfrac{\ln(n/\delta)}{\epsilon} + \tfrac{2}{\epsilon} (m^2 + m \sqrt{n}). \vspace{-0.2cm}
\end{align*}
For most parameter regimes of interest, the above bound is dominated by $2n \Delta + \frac{2}{\epsilon} (m^2 + m \sqrt{n})$.
Typically, one would expect the number of 2-stars in a graph to grow linearly, and it is upper bounded by $n \Delta^2$. Thus, our estimator is accurate when the true number of 2-stars grows at a fast-enough linear rate.

\section{Evaluation}\label{sec:eval}
\subsection{Experimental Setup}\label{sec:exp-setup}  
\noindent\textbf{Datasets.} We consider two graphs -- a real-world sparse graph and a synthetically generated dense graph.
\squishlist
    \item \textit{FB.} This graph  corresponds to data from Facebook~\cite{FB} representing 4082 users. The graph has $88$K edges. 
    \item \textit{Syn.} To test a more dense regime, we evaluate our protocols on a synthetic graph generated using the Erdos-Renyi model~\cite{ER} with parameters $G(n=4000, p=0.5)$ ($n$ is the number of edges; $p$ is the probability of including any edge in the graph).  The graph has $\approx 8$ million edges. 
\squishend

\begin{figure*}[hbt!]
\begin{subfigure}[b]{0.33\linewidth}
        \includegraphics[scale=0.5]{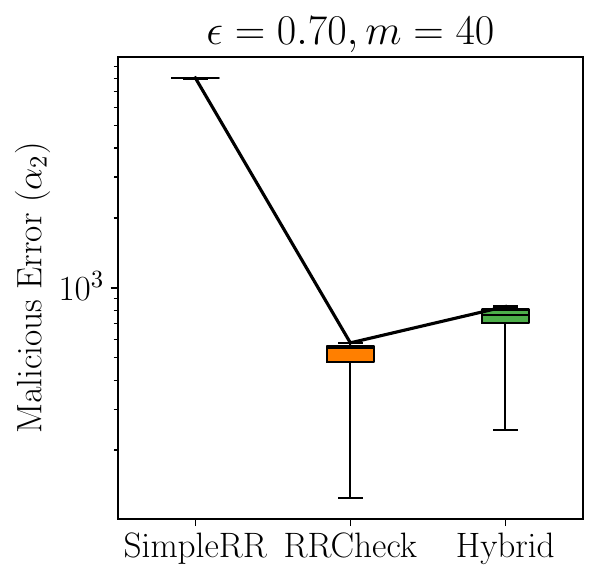}
        \caption{\textit{FB:} Degree Inflation}
        \label{fig:FB:inflation}
        \end{subfigure}
    \begin{subfigure}[b]{0.33\linewidth}
        \includegraphics[scale=0.5]{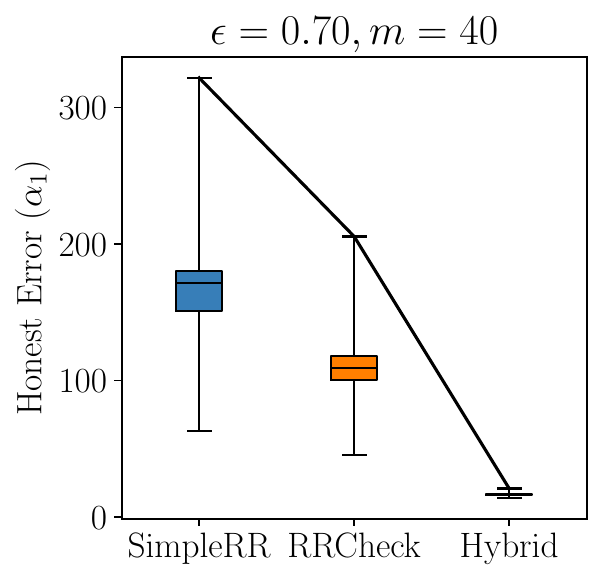}
        \caption{\textit{FB:} Degree Deflation}   \label{fig:FB:deflation}
        \end{subfigure}  
            \begin{subfigure}[b]{0.33\linewidth}
            \includegraphics[scale=0.5]{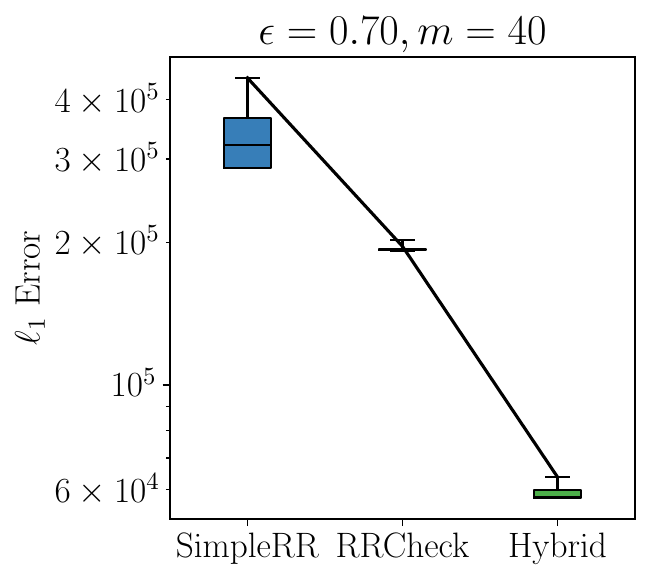}
        \caption{\textit{FB:} All}\label{fig:FB:L1}
        \end{subfigure}
         \begin{subfigure}[b]{0.33\linewidth}
            \includegraphics[scale=0.5]{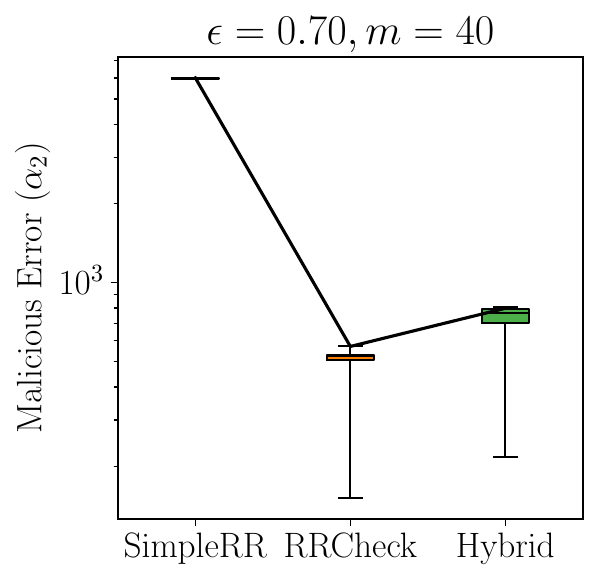}
        \caption{\textit{Syn:} Degree Inflation}\label{fig:Syn:inflation}
        \end{subfigure}
            \begin{subfigure}[b]{0.33\linewidth}
            \includegraphics[scale=0.5]{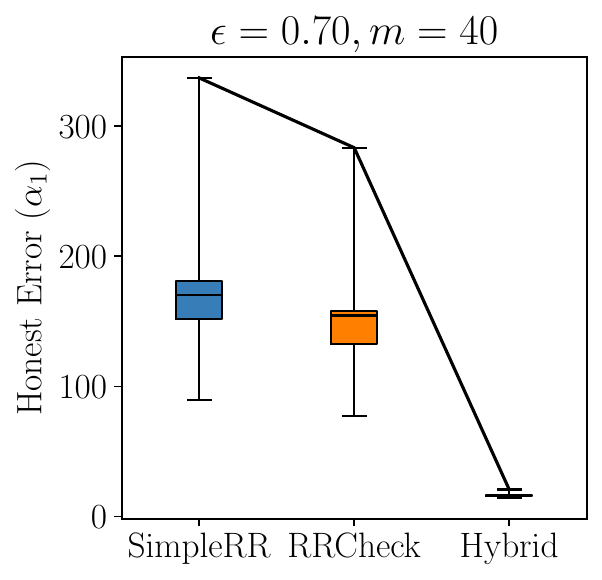}
        \caption{\textit{Syn:} Degree Deflation}\label{fig:Syn:deflation}
        \end{subfigure}
            \begin{subfigure}[b]{0.33\linewidth}
            \includegraphics[scale=0.5]{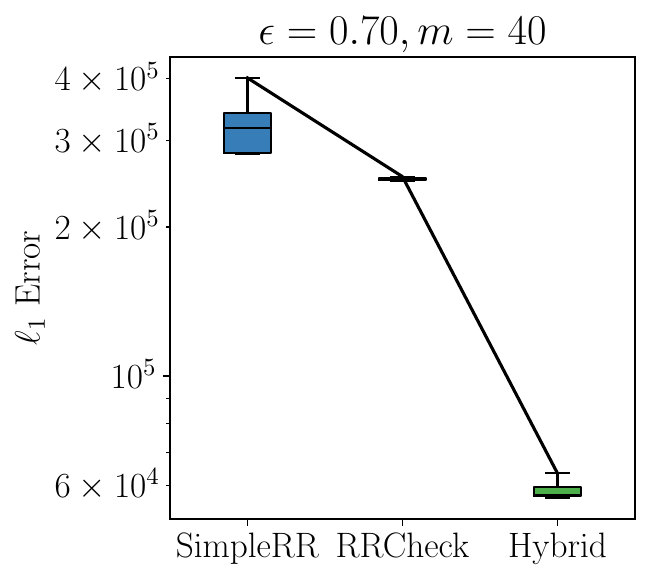}
        \caption{\textit{Syn:} All}\label{fig:Syn:L1}
        \end{subfigure}
        
    
%
%
 \small\caption{Robustness Analysis: The whiskers range from the maximum to the minimum empirical honest and malicious error observed across all the attacks of the specific type. Figs. \ref{fig:FB:L1} and \ref{fig:Syn:L1} plots the $\ell_1$ error measures of the entire degree vector across all attacks. The line corresponds to the strongest evaluated attack of the specific type.}
  \label{fig:analysis}  
\end{figure*}

\begin{figure*}[hbt!]
\centering
\includegraphics[width=0.7\linewidth]{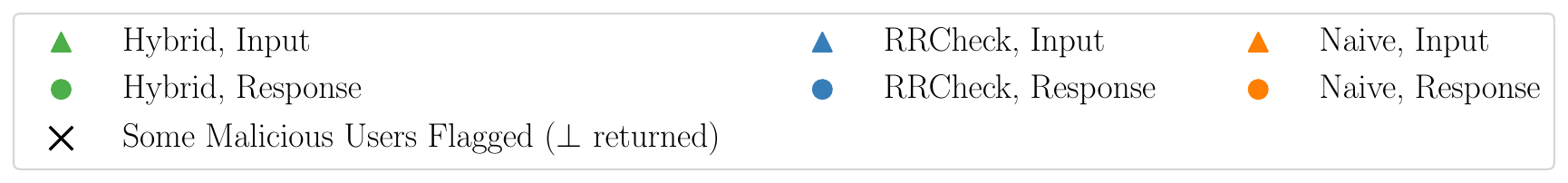}
\begin{subfigure}[b]{\linewidth}

   \includegraphics[width=0.24\linewidth]{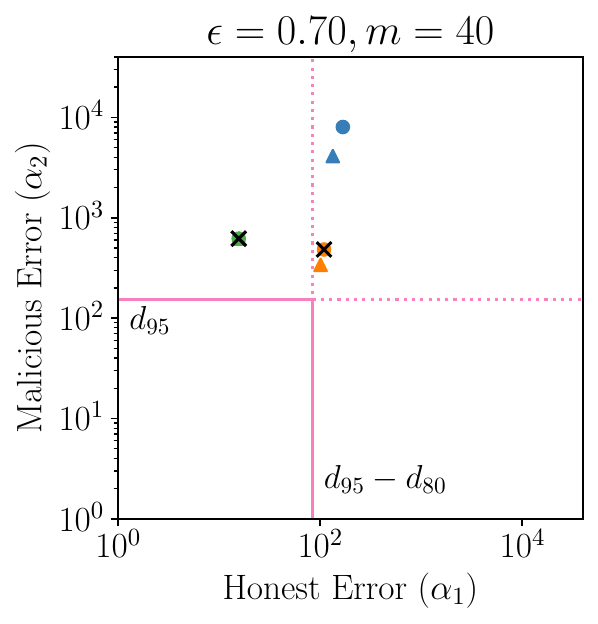}
\includegraphics[width=0.24\linewidth]{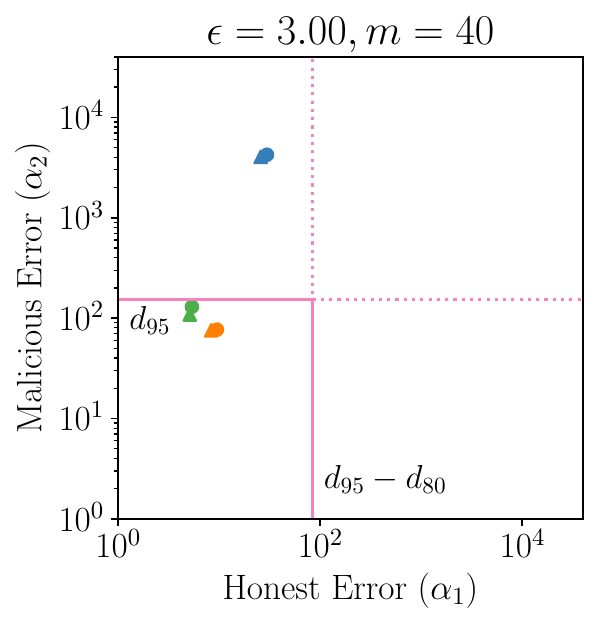}
    \includegraphics[width=0.24\linewidth]{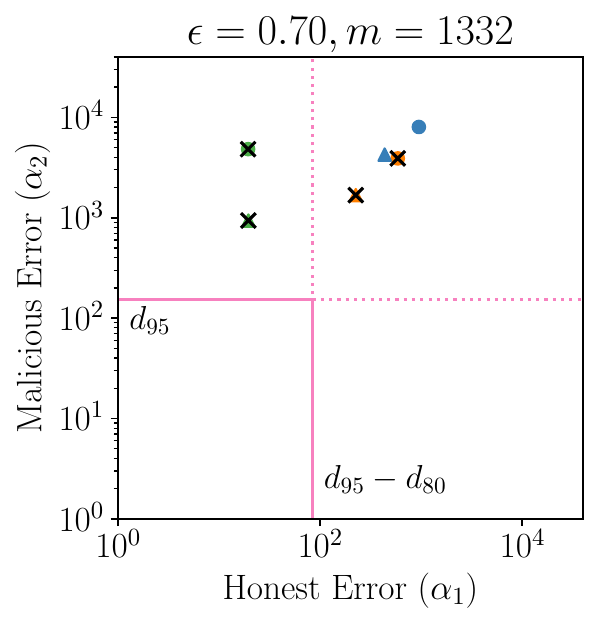}
 \includegraphics[width=0.24\linewidth]{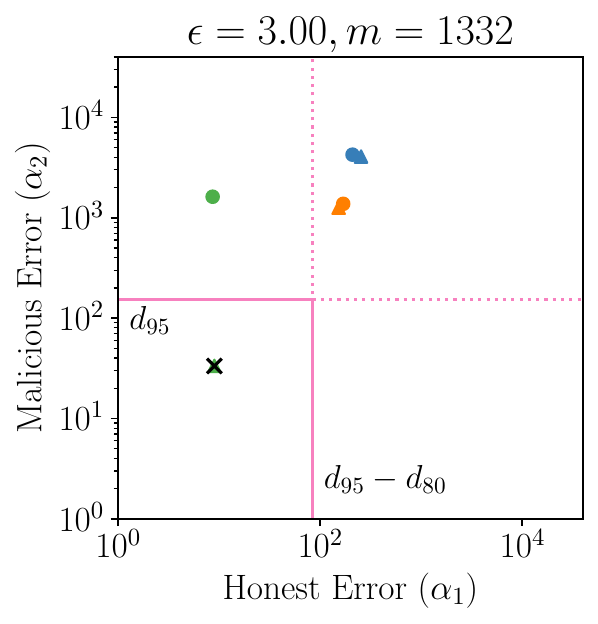}
  \caption{\textit{FB}: Combination Attack}
    \label{fig:FB:defl}
    \end{subfigure}\\
\begin{subfigure}[b]{\linewidth}
 \includegraphics[width=0.24\linewidth]{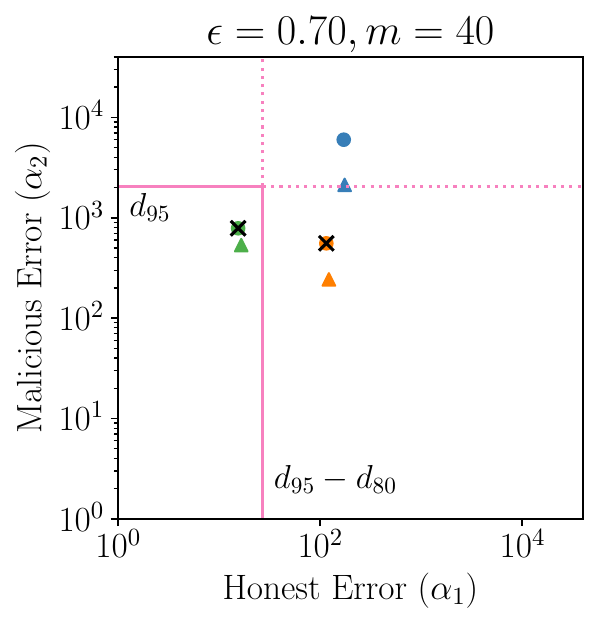}
  \includegraphics[width=0.24\linewidth]{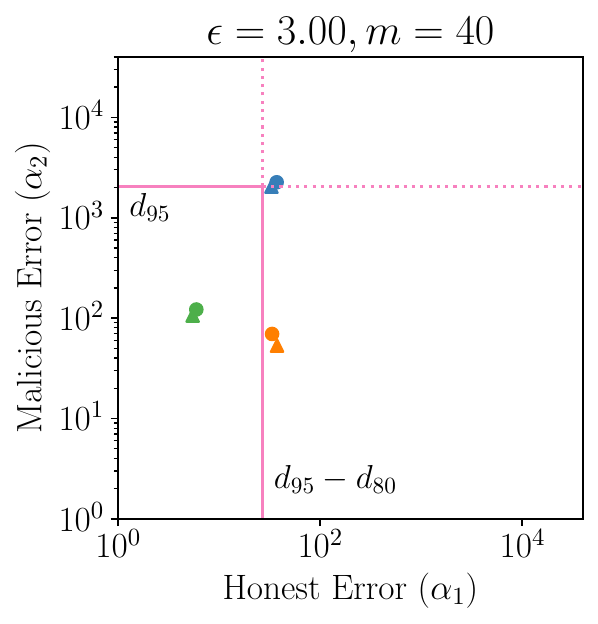}
  \includegraphics[width=0.24\linewidth]{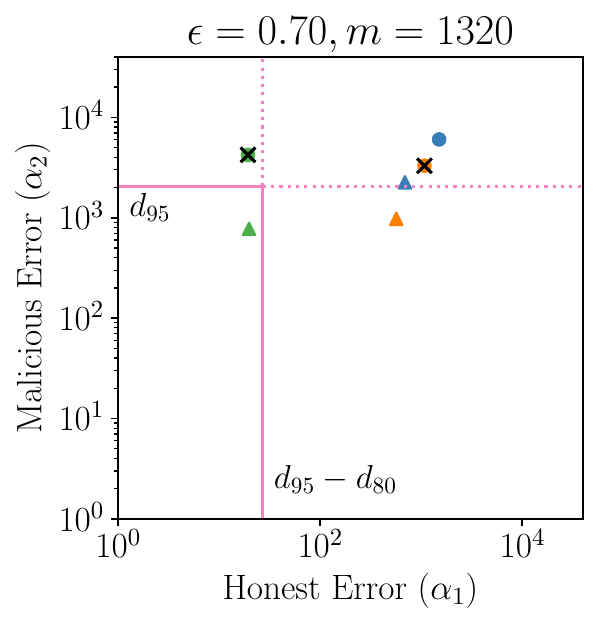}
\includegraphics[width=0.24\linewidth]{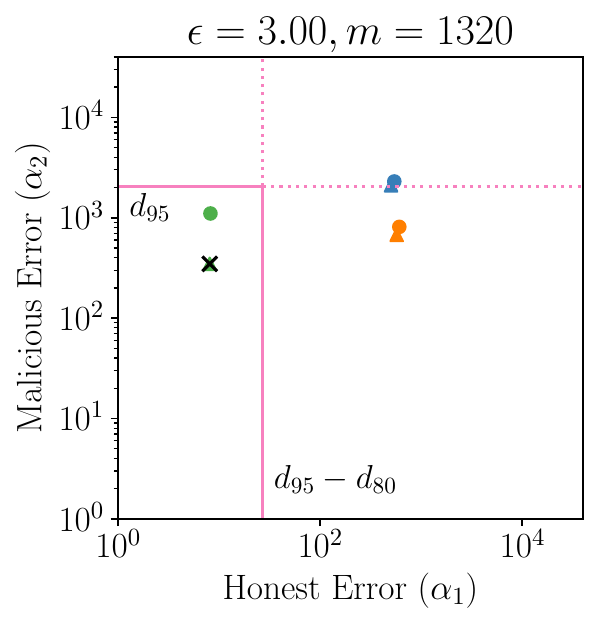}
\caption{\textit{Syn}: Combination Attack}
        \label{fig:Syn:defl}
       \end{subfigure} 

\caption{Comparison of input and response poisoning attacks: We plot the empirical accuracy (error of honest user) and soundness (error of malicious user). $d_{k}$ denotes the degree of the $k$ percentile node.}\label{fig:input} 
\end{figure*}

\noindent \textbf{Protocols.} 
We evaluate our proposed \DegRRCheck{} and \DegHybrid{}, and use \DegRRNaive{} and a defense from~\cite{Cao21} as our baselines.
\\\\
\noindent\textbf{Attacks.} We carry out an extensive analysis of the robustness of our protocols by evaluating against  \textbf{16} different attacks. The attacks are broadly classified into two types -- degree inflation and degree deflation where the goal of the malicious users is to increase (resp. decrease) the degree estimate of a target malicious (resp. honest) user by as much as possible.
 We choose these attacks because firstly, they can meet the asymptotic theoretical error bounds.
Secondly, these attacks are grounded in real-world motivations and represent practical threats (see Sec. \ref{sec:attacks}). The 16 attacks evaluated represent different configurations of the degree inflation and deflation attacks.  Specifically, they differ in $(1)$ the number of targets (both malicious or honest users), $(2)$ how the non-target malicious users are chosen 
$(3)$ collusion strategies. The different configurations capture a multitude real-world attack scenarios and adversarial goals. The attacks are summarized in Table~\ref{tab:attacks} in App. \ref{app:attacks} \ifpaper and detailed in the full paper~\cite{Fullpaper}\fi.


For each attack type, we consider both input and response poisoning versions.  In the following, let $\DO_t$ represent the target user. 
\\\noindent\DegRRCheck{}. For the degree inflation attack, the non-target malicious users always report a $1$  for the malicious user $\DO_t$ (i.e. that they are connected to $\DO_t$) in the hopes of increasing $\DO_t$'s degree estimate.
Likewise, $\DO_t$ reports $1$s for all other malicious users. For the honest users, $\DO_t$ reports extra $1$s (for non-neighbors) in the hopes of further increasing their degree estimate. The exact mechanism depends on whether it is response poisoning or input poisoning and is detailed in App. \ref{app:attacks}.
For degree deflation, we consider the worst-case scenario where $m$ of the neighbors of the honest user $\DO_t$ act maliciously. The malicious neighbors always report $0$ for their edges to $\DO_t$.
\\
\noindent\DegHybrid. 
For degree inflation, the non-target malicious users report their edges using the same strategy as in \DegRRCheck{}. For $\tilde{d}_i^{lap}$, they send their true degree estimates since their degrees are not the targets. Similarly, $\DO_t$ uses the same strategy as in $\DegRRCheck{}$ for reporting their edges. For $\tilde{d}_t^{lap}$, $\DO_t$  reports an inflated value based on the reported edges and the threshold $\tau$ (see  App. \ref{app:attacks}). For degree deflation, we consider the worst-case scenario where $m$ of the neighbors of $\DO_t$ act maliciously. The malicious users behave as they did in \DegRRCheck{} and report their true degrees for $\tilde{d}_i^{lap}$, as these are not the targeted degrees.
\\
 \noindent\DegRRNaive. For degree inflation, we consider the worst-case scenario where the target malicious user $\DO_t$ is responsible reporting all their edges, and chooses to reports all $1$s. 
For degree deflation, we again consider the worst case scenario where the malicious users are responsible for reporting the edges to $\DO_t$, and they report $0$s.
\\\\
\noindent \textbf{Configuration.} 
For every attack we report the maximum error over all the honest targets (honest error $\alpha_1$; for all experiments we take care to ensure that no honest users are mistakenly flagged) and the malicious targets (malicious error $\alpha_2$).
We run each experiment $50$ times and report the mean.  We use $\delta=10^{-6}$ and $c = 0.9$ for \DegHybrid{}. Our theoretical results suggested setting $\tau = m + C\sqrt{\rho n}$, where $C$ is a constant that is obtained from Chernoff's bounds, for the different input and response manipulation attacks. The constant $C$ is not tight, and for the practical interest of using as small a threshold as possible, we sought to set $\tau$ as small as possible so as not to falsely flag any honest user. Note that lower the threshold, lower is the permissible skew ($\alpha_1$ and $\alpha_2$  for honest  and  malicious error, respectively) introduced by poisoning, thereby improving the robustness of our protocols. We ran preliminary experiments using $50$ runs of each protocol on both graphs, and we found that at all values of $\epsilon$, setting $\tau = m + 0.4\sqrt{\rho n}$ (for $m = 40$) and $\tau = m + 0.1 \sqrt{\rho n}$ (for $m = 1500$) did not result in any false positives. Thus, we used these smaller thresholds in our experiments, and throughout the experiments there were no false positives. 

\begin{table}
    \centering
  \begin{tabular}{lcc@{\extracolsep{\fill}}cc@{\extracolsep{\fill}}cc}
    \toprule
 
        & \multicolumn{2}{c}{\textbf{\DegRRNaive{}}} & \multicolumn{2}{c}{\textbf{\DegRRCheck{}}} & \multicolumn{2}
{c}{\textbf{\DegHybrid{}}} \\ \cline{2-3} \cline{4-5}\cline{6-7}& \textit{\textbf{FB}} & \textit{\textbf{Syn}} & \textit{\textbf{FB}} & \textit{\textbf{Syn}} & \textit{\textbf{FB}} & \textit{\textbf{Syn}}\\
\midrule
        Min. & 0  & 0&$56.0\%$ & $52.0\%$ & $ \hspace{0.2cm} 54.5\%$ & $51.0\%$ \\ 
        Mean. & 0 & 0&$63.2\%$ & $61.4\%$ & $ \hspace{0.2cm} 62.1\%$ & $60.0\%$ \\ 
        Max. & 0 & 0& $75.0\%$ & $70.0\%$ & $ \hspace{0.2cm} 70.0\%$ & $70.0\%$ \\ \bottomrule
    \end{tabular}
\caption{Table of max, min, and average percentage of malicious targets flagged for degree inflation attacks.} \label{tab:flag} 
\end{table}

\subsection{Robustness Analysis} \label{sec:exp-results}

In this section, we empirically evaluate the robustness of our proposed protocols. We focus on response poisoning for every attack to analyze the worst-case scenario (since response poisoning is stronger than input poisoning). The number of malicious users is set to $m=1\%$ and $\epsilon=0.7$
. 
\\\noindent\textbf{Degree Inflation.} We report the malicious error for all the attacks that have a degree inflation component (10 out of the 16 attacks) in Figs. \ref{fig:FB:inflation} and \ref{fig:Syn:inflation} for \textit{FB} and \textit{Syn}, respectively. In other words, these are all the attacks where at least a subset of the malicious users are trying to inflate the degree of some (malicious) target.  We observe that both our proposed protocols perform significantly better than the baseline. For instance, for the strongest inflation attack we evaluated (attack A11 in Tab. \ref{tab:attacks})  -- \DegRRNaive{}  has $9.7\times$ and $13.8\times$ higher malicious error than \DegHybrid{} and \DegRRCheck, respectively, for \textit{FB}. Note that \DegRRCheck{} performs slightly better than \DegHybrid. This is because, although both the protocols have the same asymptotic soundness guarantee, the constants are higher for \DegHybrid~(see Sec. \ref{sec:hybrid}). Finally, for both datasets our protocols flag the malicious targets (returning $\bottom$) for $51\%-70\%$ of the trials (Tab. \ref{tab:flag}). This indicates that our proposed protocols are able to detect malicious users, thereby disincentivizing malicious activity. 
\\\noindent \textbf{Degree Deflation.} Figs. \ref{fig:FB:deflation} and  \ref{fig:Syn:deflation} show the results for the degree deflation attacks on  \textit{FB} and \textit{Syn}, respectively. Specifically, we report the honest error for all the attacks that have a (11  out of the 16 attacks evaluated). We observe that \DegHybrid{} performs the best. For instance, for the strongest deflation attack we evaluated (attack A8 in Tab. \ref{tab:attacks}) it has $16.2\times$ and $13.6\times$ better accuracy than \DegRRNaive{} and \DegRRCheck{}, respectively, for \textit{Syn}. Additionally, our protocols are able to flag malicious users when they target a large number of honest users. Specifically, for the strongest degree deflation attack, \DegHybrid{} flags  $4.5\%$  and $49.8\%$ of the malicious users for \textit{FB} and \textit{Syn}, respectively. \DegRRCheck{}, on the other hand, flags $3\%$ and $59.3\%$  of the malicious users for \textit{FB} and \textit{Syn}, respectively. Note that the number of actual honest users affected by a malicious user is bounded by its degree. Hence, the rate of flagging is less aggressive for \textit{FB} since it is a sparse graph. 
\\\noindent \textbf{$\boldsymbol{\ell_1}$ Error.} We report the $\ell_1$ error of the entire noisy degree vector  $\mathbf{\hat{d}}=\langle \hat{d}_1, \ldots, \hat{d}_n\rangle$ in Fig. \ref{fig:FB:L1}  and \ref{fig:Syn:L1}, for \textit{FB} and \textit{Syn}, respectively. We observe that \DegHybrid{} performs the best. This is because recall that \DegHybrid{} has the best honest error guarantee while its malicious error is comparable to that of \DegRRCheck{}.  Since the number of honest users is much higher than the number of malicious users, the $\ell_1$ error of \DegHybrid{} is significantly better than that of \DegRRCheck{} due to its better honest error. For instance, for the strongest overall attack we evaluated in terms of the $\ell_1$ error (A8 -- the same as that for the degree inflation case) \DegHybrid{} has $4.0\times$ and $6.3\times$ lower $\ell_1$ error than \DegRRCheck{} and \DegRRNaive, respectively for \textit{Syn}.
\\\noindent\textbf{The Effect of Collusion.} We empirically evaluate the efficacy of our proposed heuristics defenses against strong collusion (Sec. \ref{sec:discussion}). We consider an attack in which all $m$ malicious users collude to inflate the degree estimates of an entire subgraph (i.e., they form a clique). 
This results in highly correlated adjacency lists which can be detected, even after perturbing them under \ldp. To evaluate this, we run a greedy, correlation-based clique detection algorithm and report the minimum value of $m$ for which the injected clique is detected at least once across 4 trials (Table~\ref{tab:clique-detection:1}). What this means is that any collusion exceeding the reported values is almost certain to be detected. For instance, in the \textit{Syn} dataset, collusion involving more than 2.5\% of the users is reliably detected.

 \subsection{Comparing Input and Response Poisoning}
We plot the efficacies of input and response poisoning attacks in Fig. \ref{fig:input}. For this, we choose the strongest combination attack we evaluated (A10 in Tab. \ref{tab:attacks}). This attack considers the same number of honest and malicious targets and all the malicious users perform degree inflation and deflation simultaneously. All the users (honest targets and malicious) are selected from the same community in the graph (detected via the modularity maximization community detector~\cite{clauset2004finding}). This ensures that the malicious users have true edges with most of the honest targets which increases the potency of the deflation attack. 
  We consider two values of $m$ (number of malicious users): $m=1\%$ and $m=33\%$. While $m=1\%$ represents a realistic threat model, we also consider $m=33\%$ to showcase the efficacy of our protocols even with a large number of malicious users.
  \par We observe that input poisoning is weaker than response poisoning in terms of both honest and malicious error. Specifically, for malicious error it is worse than response poisoning for all three protocols (since response poisoning has an extra $O(\frac{m}{\epsilon}$) term). As a concrete example, for \DegHybrid, response poisoning is $2.5\times$ worse than input poisoning for \textit{FB} with $m=33\%,\epsilon=0.7$. In terms of honest error, input poisoning is worse than response poisoning (again because of the extra $O(\frac{m}{\epsilon}$) term in response poisoning). For instance, for \DegRRCheck{}, response poisoning is $2.7\times$ worse than input poisoning for \textit{Syn} with 
  $m=33\%,\epsilon=0.7$. However, the accuracy for \DegHybrid{} is comparable for both input and response poisoning which is consistent with our theoretical results (Thm. \ref{thm:input:hybrid}). This is because under both types of attacks, \DegHybrid{} uses the honest users' Laplace estimates which are not affected by the malicious users.  As expected, the separation between input and response poisoning becomes less prominent with higher $\epsilon$ and lower $m$, as it is harder to pull off strong attacks for these regimes.

\par We also mark the degree of $95^{th}$ percentile node ($d_{95}$) for the respective graphs in the plots. The way to interpret this is as follows.  If an error of a protocol falls below the line, then any malicious user can inflate their degree estimate to be in the $>d_{95}$ percentile by staging a poisoning attack. Our protocols are more performant for the dense graph \textit{Syn} (attacks are prevented for both values of $\epsilon$). This is because of the $\tilde{O}(\frac{\sqrt{n}}{\epsilon})$ term in the malicious error --- is the dominant term for sparse graphs. 

\par We plot the measure $d_{95}-d_{80}$ in Figs. \ref{fig:FB:defl} and \ref{fig:Syn:defl} where $d_{k}$ denotes the degree of the $k$ percentile node. The way to interpret this is as follows. If an error falls to the left of the line, then malicious users can successfully deflate the degree of an honest target from $>95$ percentile to $<80$ percentile. Based on our results, we observe that \DegHybrid{} is mostly effective in protecting against this attack even with a large number of malicious users of $m=33\%$. 

\subsection{Impact of Algorithmic Parameters}\label{sec:eval:param}

\begin{figure*}
    \centering\includegraphics[width=0.7\linewidth]{Plots_new/legend.pdf} \\
  \begin{subfigure}[b]{0.24\linewidth}
    \centering \includegraphics[width=\linewidth]{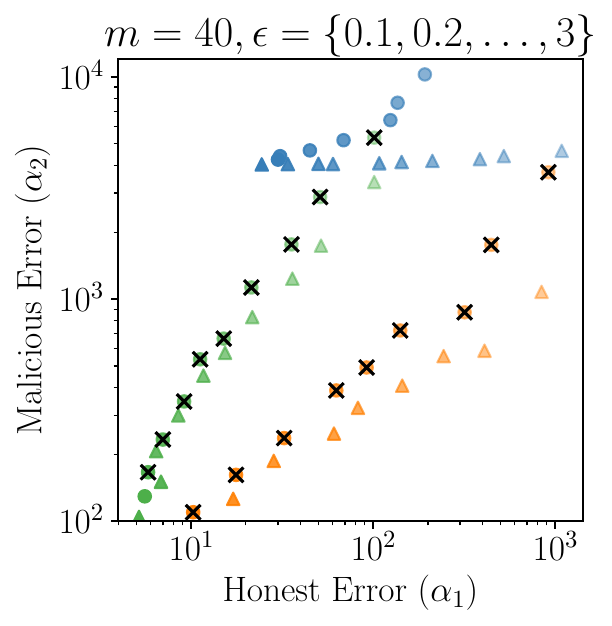}  
    \caption{\scalebox{0.8}{\textit{FB}: Combination}}
        \label{fig:FB:privacy}\end{subfigure}
        \begin{subfigure}[b]{0.24\linewidth}
        \includegraphics[width=\linewidth]{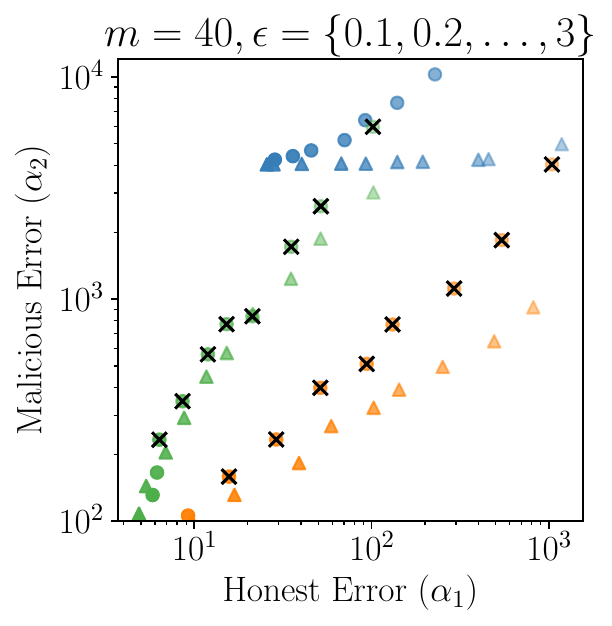} \caption{\textit{Syn}: Combination}
        \label{fig:Syn:privacy}
        \end{subfigure}
        \begin{subfigure}[b]{0.24\linewidth}
        \centering
        \includegraphics[width=\linewidth]{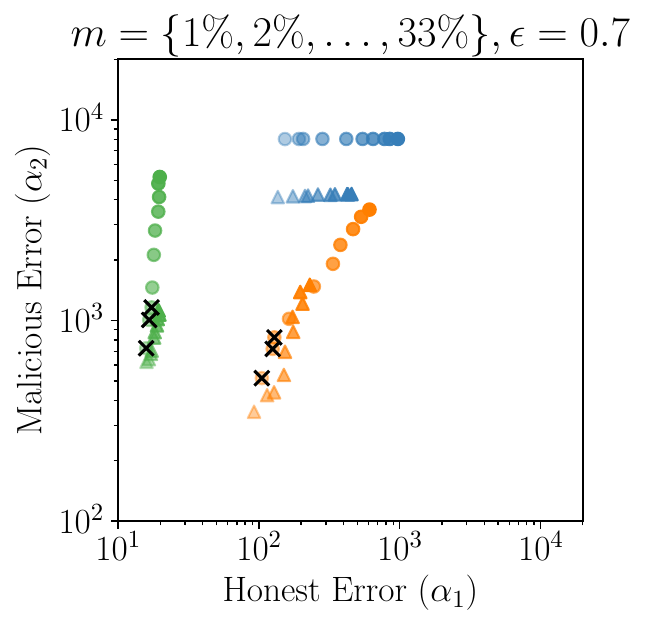}
        \caption{\textit{FB}: Combination}
        \label{fig:FB:m}
        \end{subfigure} 
         \begin{subfigure}[b]{0.24\linewidth}
    \centering \includegraphics[width=\linewidth]{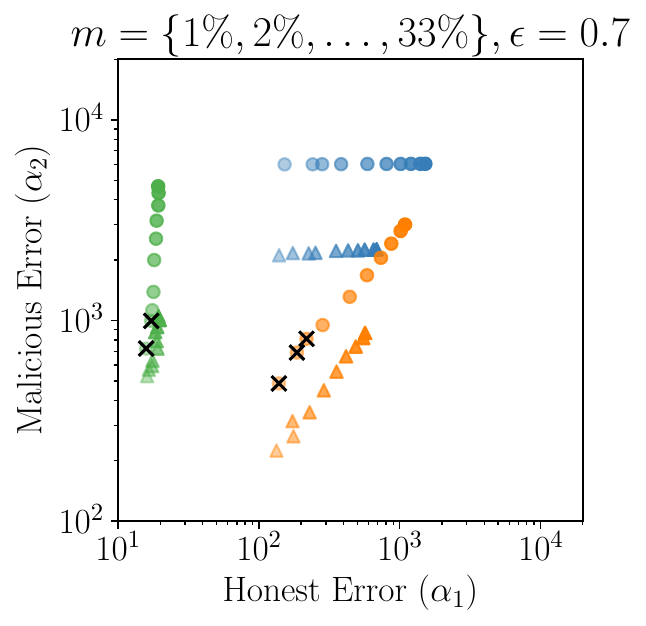}  \caption{\textit{Syn}: Combination}
        \label{fig:Syn:m}\end{subfigure}

\caption{Robustness analysis with varying $\epsilon$ (Figs. \ref{fig:FB:privacy}, \ref{fig:Syn:privacy}) and $m$ (Figs. \ref{fig:FB:m}, \ref{fig:Syn:m}): Higher brightness denotes higher $\epsilon$ ($m$).}  
\end{figure*}

\begin{table}
\centering
\begin{tabular}{lllll}
\toprule
    & \multicolumn{2}{c}{$\mathbf{\epsilon = 0.7}$} &\multicolumn{2}{c}{$\mathbf{\epsilon=3.0}$} \\ 
    & \textbf{\scalebox{0.8}{Input Poisoning}} & \textbf{\scalebox{0.8}{Response Poisoning}} & \textbf{\scalebox{0.8}{Input Poisoning}} & \textbf{\scalebox{0.8}{Response Poisoning}} \\ \midrule
 \textit{FB} & 13\% & 3.7\% & 3\% & 2.5\% \\ 
\textit{Syn} & 38.6\% & 10.4\% & 2.7\% & 2.4\% \\ \bottomrule
\end{tabular}
\caption{Minimum percentage of malicious users such that a strong collusion of all the malicious users was caught on all four poisoning attempts.}\label{tab:clique-detection:1}
\end{table}

\begin{figure*}[hbt!]
\begin{subfigure}[b]{0.33\linewidth}
        \includegraphics[scale=0.5]{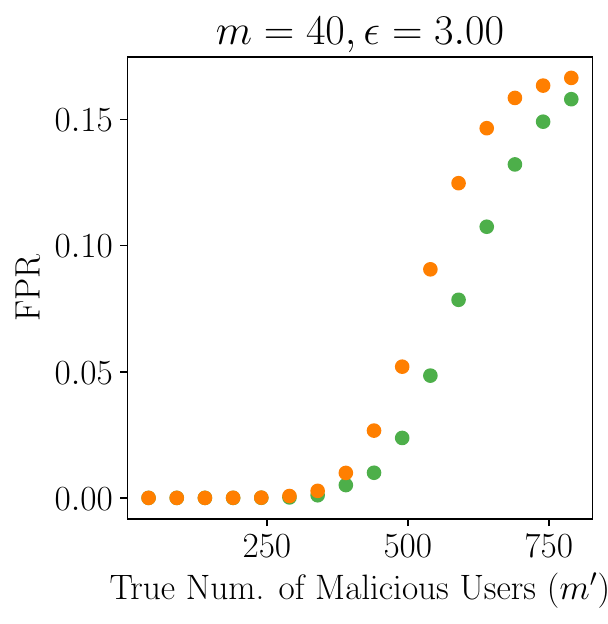}
        \caption{\textit{FB:} Degree Deflation}
        \label{fig:FB:FPR}
        \end{subfigure}
    \begin{subfigure}[b]{0.33\linewidth}
        \includegraphics[scale=0.5]{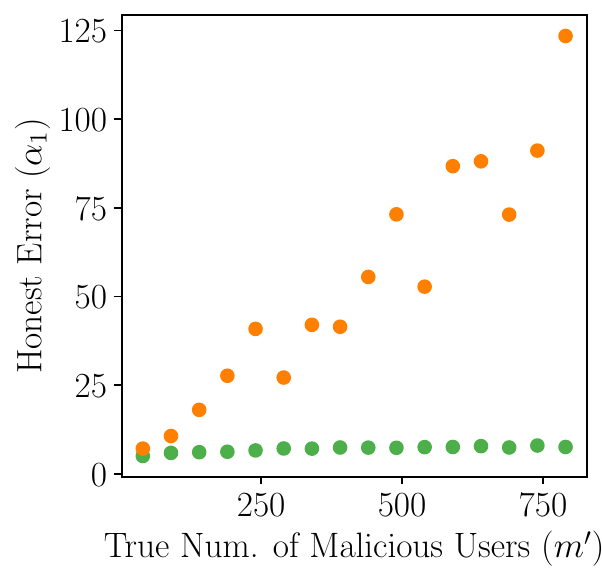}
        \caption{\textit{FB:} Degree Deflation } \label{fig:FB:honest_error}
        \end{subfigure}  
            \begin{subfigure}[b]{0.33\linewidth}
            \includegraphics[scale=0.5]{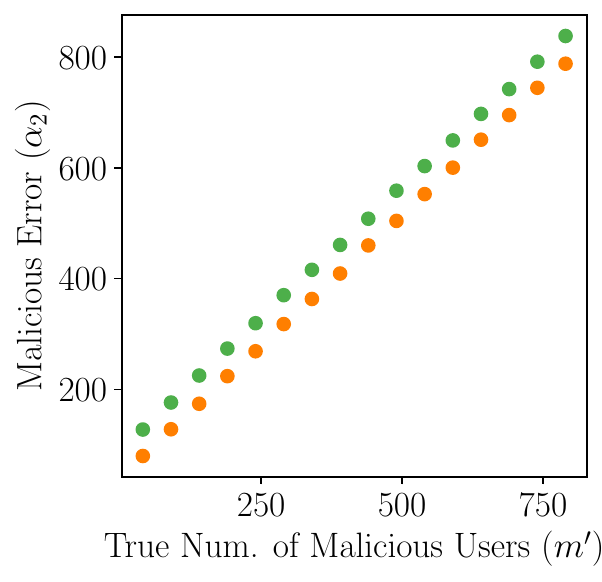}
            \caption{\textit{FB:} Degree Inflation }\label{fig:FB:malicious_error}
        \end{subfigure}
         \begin{subfigure}[b]{0.33\linewidth}
            \includegraphics[scale=0.5]{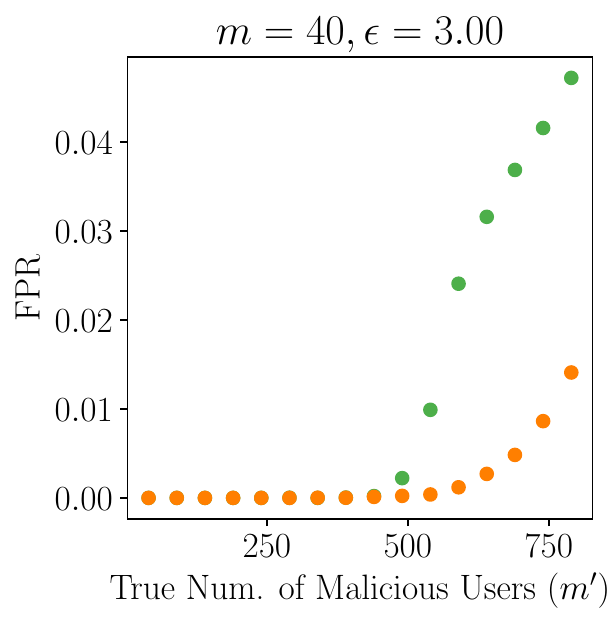}
            \caption{\textit{Syn:} Degree Deflation }\label{fig:Syn:FPR}
        \end{subfigure}
            \begin{subfigure}[b]{0.33\linewidth}
            \includegraphics[scale=0.5]{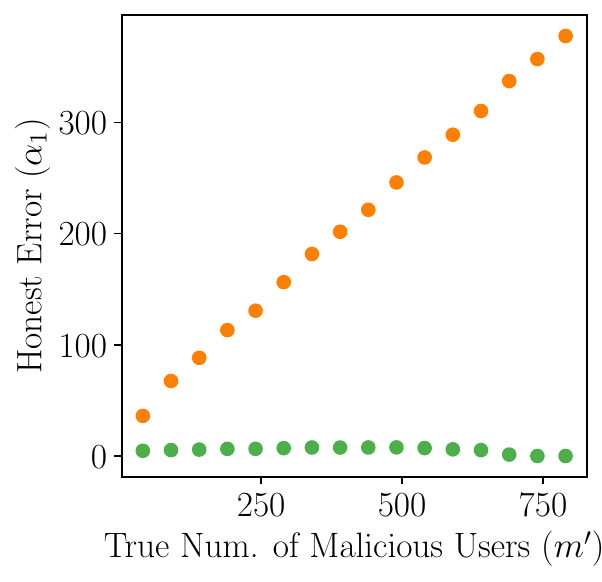}
            \caption{\textit{Syn:} Degree Deflation }\label{fig:Syn:honest_error}
        \end{subfigure}
            \begin{subfigure}[b]{0.33\linewidth}
            \includegraphics[scale=0.5]{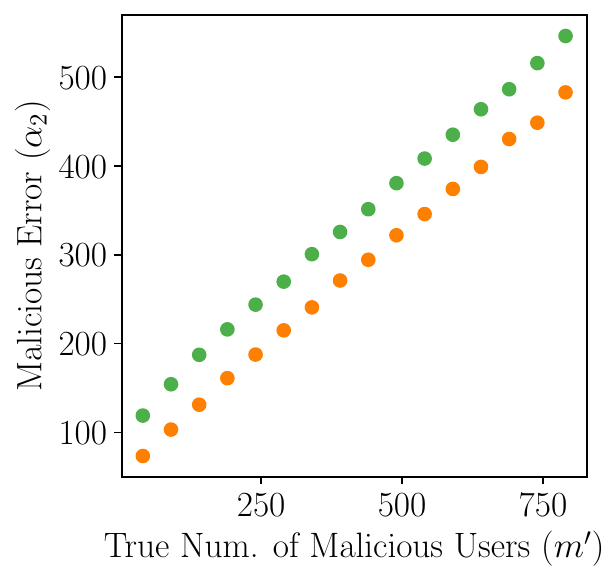}
            \caption{\textit{Syn:} Degree Inflation}\label{fig:Syn:malicious_error}
        \end{subfigure} 
 \small\caption{Robustness analysis when the true number of malicious users exceeding the theoretically assumed bound.}
  \label{fig:super:1} 
\end{figure*}

\noindent\textbf{The Effect of $\boldsymbol{\epsilon}$.} \ifpaper Fig. \ref{fig:FB:privacy} shows \else Figs. \ref{fig:FB:privacy} and \ref{fig:Syn:privacy} show\fi the impact of the attacks with varying privacy parameter $\epsilon$. We study the strongest combination attack (attack A10) which considered in the previous section.  We observe that, increasing privacy (lower $\epsilon$) leads to more skew for all attacks on all three protocols. For instance, the soundness of the response poisoning version of the degree deflation attack for \textit{FB}  $42\times$ worse for $\epsilon=0.1$ than that for $\epsilon=3$ for \DegHybrid{}. Additionally, we observe that malicious users get flagged only response poisoning since this is a stronger attack than input poisoning. 
\\\noindent\textbf{The Effect of $\boldsymbol{m}$.} 
We show how the attack efficacy varies with $m$ in \ifpaper Fig. \ref{fig:FB:m}\else Figs. \ref{fig:FB:m} and \ref{fig:Syn:m}\fi. We consider the strongest combination attack (attack A10) here as well. As expected, the impact of poisoning worsens with increasing $m$. Specifically, both honest and malicious error of \DegRRCheck{} worsen with increasing $m$. For instance, the honest error of response poisoning degrades by $8.5\times$ as we increase $m$ from $1\%$ to $33\%$ for \textit{Syn}. The honest error is uninfluenced by $m$ for \DegHybrid{} -- this is because it reports the Laplace estimates for the honest users which are not impacted by the malicious users. On the other hand, the malicious error remains unaffected for \DegRRNaive{}~since here a malicious target can always carry-out the worst-case attack regardless of $m$\ifpaper \else (Thm. \ref{thm:response:naive})\fi. 

We also consider the case when the true number of malicious users ($m'$) exceeds the assumed bound $m$. First, we consider the degree deflation attack, where the malicious users are colluding to deflate the degree of an honest target user $\DO_t$.  In \DegHybrid{}, recall that the final degree estimate—after the consistency checks—is obtained using the Laplace mechanism. Now let's take the case of the honest user $\DO_t$'s Laplace-based degree estimate in \DegHybrid{} -- since this value is completely controlled by the honest user, the malicious users \textit{cannot} increase the honest error, even when their number exceeds the assumed bound. This is demonstrated in Figs.~\ref{fig:FB:honest_error} and \ref{fig:Syn:honest_error}.  However, malicious users can falsely report their common edges with the target honest user $\DO_t$, thereby increasing the inconsistency count and causing the protocol’s consistency checks to fail. As a result, the honest user may be mistakenly flagged as malicious as demonstrated in Figs.~\ref{fig:FB:FPR} and \ref{fig:Syn:FPR}. Next we consider a degree inflation attack, where malicious users collude to inflate the degree of a target malicious user $\DO_t$. The error in $\DO_t$'s degree estimate can be increased in two ways -- (1) by falsifying edges with honest users, and (2) by having all malicious users lie consistently about their mutual edges with $\DO_t$.  The first strategy is still \textit{prevented} by our protocols, since the threshold $\tau$ used in the consistency checks is computed with respect to the lower value $m$. As a result, if $\DO_t$ attempts to report more inconsistent edges than allowed under this threshold, they will be flagged. Therefore, only the second strategy remains viable in our protocols. As a result, the malicious error in the $\DO_t$'s degree estimate increases  linearly with the number of colluding malicious users (Figs.~\ref{fig:FB:malicious_error} and \ref{fig:Syn:malicious_error}). However, as demonstrated in Table \ref{tab:clique-detection:1}, such a strong collusion can be effectively caught by our heuristic defenses.

\subsection{Comparison with Prior Work}
We present experiments comparing our protocols with the \textit{only} compatible prior defense by Cao et al.\cite{Cao21} -- originally designed for private frequency estimation (computing a histogram) of tabular data. Their defense flags users reporting suspiciously frequent target items. While their setting differs from ours (unlike a histogram, degree vector $\mathbf{\hat{d}}=\langle \hat{d}_1, \ldots, \hat{d}_n\rangle$ is \textit{not} a aggregate query  -- it is a \textit{per-user} degree estimate), their binary input format aligns with our adjacency lists, making adaptation possible. Other prior defenses are structurally incompatible, as discussed in Sec.\ref{sec:relatedwork}.
\\\noindent\textbf{Degree Inflation Attack.} Intuitively, both our protocols and the baseline defense rely on a threshold parameter to govern the flagging decision. Lowering  the threshold increases the likelihood of catching malicious users but may also raise the number of honest users that are incorrectly flagged, revealing a natural trade-off between the true positive rate (TPR) and the false positive rate (FPR). We empirically evaluate this trade-off by carrying out the strongest inflation attack (A11 from Table~\ref{tab:attacks}) in Fig.\ref{fig:FB:tau}\ifpaper \else and Fig.\ref{fig:Syn:tau}\fi. Our results show that our protocol achieves high TPR while maintaining low FPR, demonstrating its effectiveness. Furthermore, the plots clearly delineate the expected trade-off between TPR and FPR.
In contrast, the baseline defense is completely ineffective -- it fails to flag any malicious user and essentially flags honest users at random. It also exhibits no meaningful trade-off between TPR and FPR.
\\\noindent\textbf{Adaptive Attack.} There is an inherent trade-off for the adversary in our protocols -- they must balance between maximizing the error introduced and minimizing the likelihood of being flagged. Based on the knowledge of $\tau$, the adversary can construct the best possible adaptive attack by optimizing for the two objectives - this is the theoretically optimal adaptive attack. 
 We model this optimal adaptive attack and illustrate it in \ifpaper Fig. \ref{fig:FB:adaptive:1}\else Figs. \ref{fig:FB:adaptive:1} and \ref{fig:Syn:adaptive:1}\fi, where the attacker adaptively chooses how many bits to falsify based on $\tau$ during a degree inflation attack. Each point shows the trade-off between error introduced and the flagging risk. Our results reveal a clear \textit{Pareto frontier} -- more falsification increases both error and detection risk. In contrast, \DegRRNaive{} with the baseline defense from \cite{Cao21} provides no protection as the malicious users can increase error arbitrarily without being flagged.

\begin{figure*}
    \centering\includegraphics[width=0.5\linewidth]{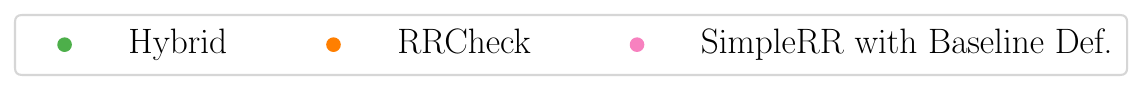} \\
  \begin{subfigure}[b]{0.24\linewidth}
\includegraphics[width=\linewidth]{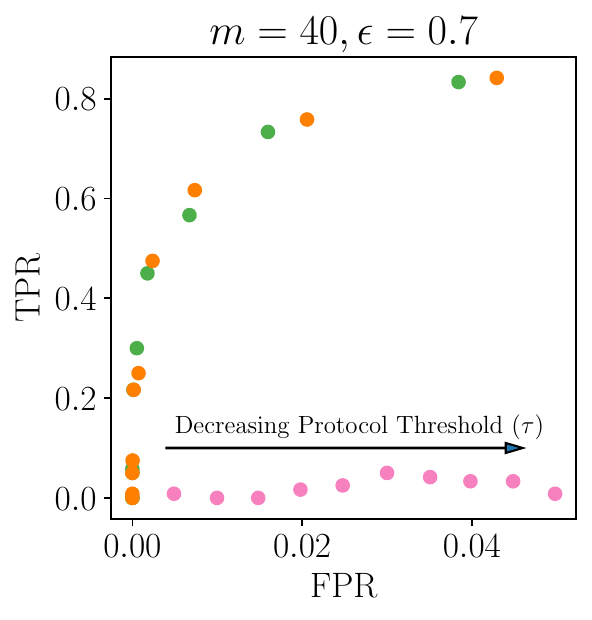}
        \caption{\textit{FB}: Degree Inflation }
        \label{fig:FB:tau}\end{subfigure}
        \begin{subfigure}[b]{0.24\linewidth}
        \includegraphics[width=\linewidth]{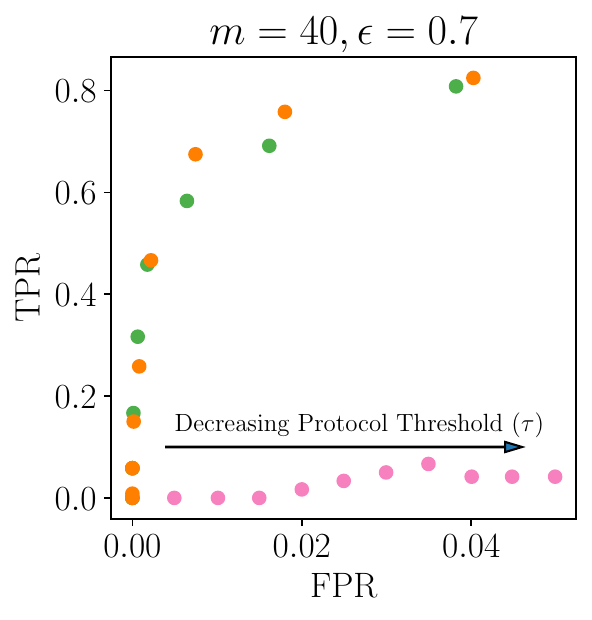}
    \caption{\textit{Syn}: Degree Inflation}
        \label{fig:Syn:tau}
        \end{subfigure}
        \begin{subfigure}[b]{0.24\linewidth}
    \centering \includegraphics[width=\linewidth]{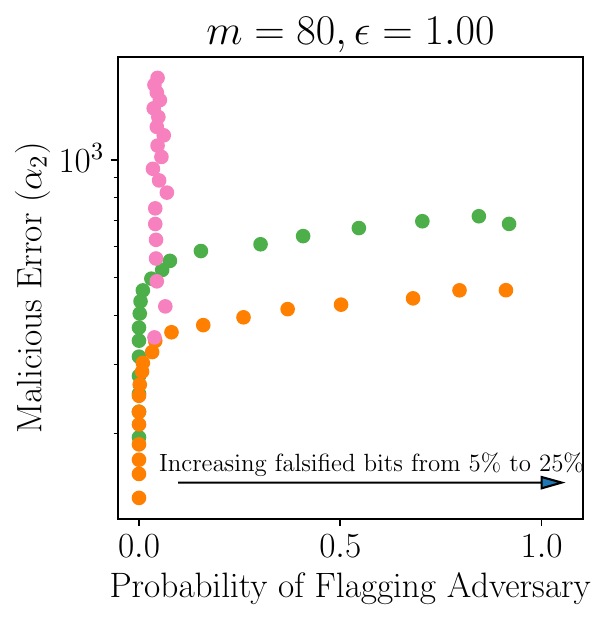}
      \caption{\textit{FB}: Adaptive }
        \label{fig:FB:adaptive:1}\end{subfigure}
        \begin{subfigure}[b]{0.24\linewidth}
        \includegraphics[width=\linewidth]{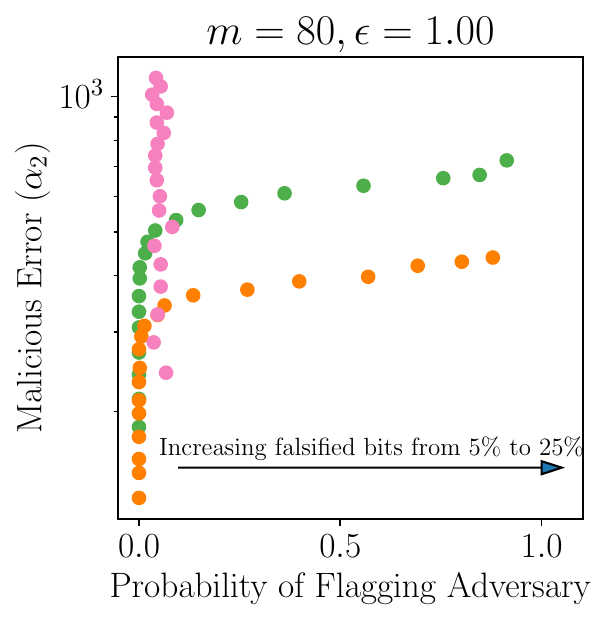}
 \caption{\textit{Syn}: Adaptive}
        \label{fig:Syn:adaptive:1}
        \end{subfigure}
\small{\caption{Comparison with the baseline from \cite{Cao21}. Figs. \ref{fig:FB:tau} and \ref{fig:Syn:tau} show the efficacy of flagging malicious users while Figs. \ref{fig:FB:adaptive:1} and \ref{fig:Syn:adaptive:1} show the \textit{Pareto frontier} for adaptive attacks.}}
\end{figure*}

\section{Related Work}\label{sec:relatedwork}
\textit{Poisoning Attacks for LDP Protocols} A recent line of work~\cite{Cheu21,Cao21,Wu21,KV1,Li22} has explored the impact of poisoning in the \ldp~setting. However, these works focused either on tabular data or key-value data. Additionally, prior work mostly focuses on the task of frequency estimation which is different from our problem of degree estimation. For the former, each user has some item (or (key,value) pair) from an input domain  and the data aggregator wants to compute the histogram over all the users’ items. Whereas, we compute the degree vector $\langle \hat{d}_1, \ldots, \hat{d}_n\rangle$ -- each user directly reports their degree $d_i$ (a count or via an adjacency list). 
A few of the works~\cite{Cao21,Wu21, 10415225} also presented some ad-hoc defenses against their proposed attacks. However, these defenses either protect against specific attacks, make assumptions about the data distribution or are interactive (recall we focus on non-interactive protocols). 
\\\textit{Cryptographic Protection.}
Prior work has also explored cryptographic strategies to defend against poisoning in the context of \ldp~\cite{Kato21, Ambainis03, Moran06, 10.1007/978-3-031-68208-7_18}. These approaches typically focus on verifying the correctness of the \ldp, ensuring that users execute the randomization mechanism as specified. However, they do \textit{not} protect against input poisoning. In contrast, our protocols are designed to defend against input poisoning  as well (see Sec.\ref{sec:input-attacks}). While recent work~\cite{bontekoe2024efficientverifiabledifferentialprivacy,Verity} also addresses input poisoning, it assumes that inputs are provided by a trusted third-party authorizer -- an assumption that does not hold in our setting of fully distributed graphs which has no authorizer.
\\\textit{Poisoning Attacks in Other Context.} A slew of poisoning attacks~\cite{biggio2021poisoning,mei2015teaching,Fang2020LocalMP,bhagoji2019analyzing, chen2017targeted,bagdasaryan2018backdoor,Xie2020DBADB} and defenses~\cite{roychowdhury2021eiffel, burkhalter2021rofl} for ML models have been proposed~\cite{kairouz2019advances}. 
These attacks are fundamentally different from the ones in our setting. For ML, the users send parameter gradient updates (multi-dimensional real-valued vector) and the attack objective is to misclassify data. 
Hence, none of the techniques from this literature are directly applicable here.

\section{Conclusion}\label{sec:conclusion}
In this paper, we have studied the impact of poisoning attacks on degree estimation under \ldp~and introduced a formal framework to quantify their effect on honest and malicious users. Additionally, we have proposed novel robust degree estimation protocols under \ldp~by leveraging the natural data redundancy in graphs, that can significantly reduce the impact of poisoning attacks.
\\\noindent \textbf{Acknowlegements.} JI partially carried out this work at
Basic Algorithms Research Copenhagen (BARC), which was supported by the VILLUM Investigator grant 54451. JI and KC would also like to acknowledge NSF grant CNS 2241100 for research support. ARC would like to acknowledge CI fellowship for research support.

\bibliographystyle{ACM-Reference-Format}
\bibliography{references}

\appendix

\section{Impact of Poisoning on Baseline Protocols}\label{app:baselines}
Within our robustness framework, we analyze the two naive private mechanisms outlined in Sec.~\ref{sec:ldp} -- the Laplace mechanism and randomized response. The shortcomings of these mechanisms motivate the design of our robust protocols discussed in the paper. 
\subsection{Laplace Mechanism} 

The simplest mechanism for estimating an user's degree is the Laplace mechanism, \RLap, where each user directly reports their noisy estimates. Consequently, the degree estimate of an honest user \textit{cannot} be tampered with at all -- the $\tilde{O}(\frac{1}{\epsilon})$\footnote{$\tilde{O}$ hides factors of $\log\frac{1}{\delta}$ } term is due to the error of the added Laplace noise. This error is in fact optimal (matches that of central \DP) for degree estimation.  On the flip side, a malicious user can flagrantly lie about their estimate without detection resulting in the the worst-case malicious error. Specifically,  there
exists a graph and an attack against \RLap{} in
which a malicious user is guaranteed to manipulate their true degree by $n -1$ --
this holds for the case where the malicious user is an isolated node but
lies that their degree is $n - 1$. The robustness of \RLap{} against response poisoning attacks is formalized as follows: 
\begin{thm}\label{thm:response:laplace}
	Under response poisoning, The \RLap{} protocol achieves $\frac{1}{\epsilon}\log\frac{n}{\delta}$-honest error. However, there is a response poisoning attack $\calA$ and a graph $G$ such that $\Pr[\exists i.\mathrm{err}^{mal}(d_i, \tilde{d}_i) = n-1] = 1$.
\end{thm}
The proof is in \ifpaper the full paper \cite{Full} \else  App. \ref{app:thm:response:laplace}\fi.
Thus, while \RLap{} has low honest error, it does not achieve $\alpha$-malicious error for any $\alpha \leq n-1$. Intuitively, this is because there is no way to verify the malicious users' reports. 
It is important to note that  \RLap{} has low honest error even with $n-1$ malicious users while the attack in Theorem~\ref{thm:response:laplace} is possible with just a single malicious user.

\subsection{Randomized Response}\label{sec:protocol:naive}
In this section, we look at an alternative mechanism where the users release their edges via randomized response. Recall that the information about an edge is shared between two users -- the idea here is to leverage this \textit{distributed information}. For our baseline algorithm, \DegRRNaive~(described in Alg.~\ref{alg:degrrnaive}), the data aggregator collects information about an edge from a \textit{single} user. Specifically, for edge $(i,j)$ with $i < j$, it simply uses the response from user $\DO_i$ to decide if the edge exists. To estimate the degree, it counts the total number of edges to user $\DO_i$ with the random variable $count_i^1$ and then computes a debiased estimate of the degree. Note that this naive approach is used by many prior works in graph algorithms~\cite{LDPGraph1, LDPGraph2,imola2021locally,imola_communication-efficient_2022}.
\setlength{\textfloatsep}{4pt}
Formally for response poisoning attacks, we have:
\begin{thm}\label{thm:b3a2_hard} 
Under response poisoning \DegRRNaive{} protocol attains $m\frac{e^\epsilon+1}{e^\epsilon-1}+\sqrt{n}\frac{\sqrt{(e^\epsilon+1)\ln\frac{2n}{\delta}}}{e^\epsilon-1}$ honest error. However, there is a response poisoning attack $\calA$ and a graph $G$ such that $\Pr[\exists i.~\mathrm{err}^{mal}(d_i, \hat{d}_i) = n-1] = 1$.
\label{thm:response:naive}\end{thm}
The above theorem is proved in \ifpaper the full paper \cite{Full} \else  App. \ref{app:thm:response:naive}\fi.
For $\epsilon < 1$, the guarantee for honest error is $\approx m(1+\frac{1}{\epsilon})+\frac{\sqrt{n}}{\epsilon}$. Intuitively, the $\frac{\sqrt{n}}{\epsilon}$ term comes from the error introduced by randomized response. The $m(1+\frac{1}{\epsilon})$ term comes from the adversarial behavior of the malicious users -- $m$ term is inevitable and 
 accounts for the worst case scenario where all $m$ malicious users are colluding, while the $\frac{1}{\epsilon}$ factor corresponds to the scaling factor required for de-biasing. This observation is in line with prior work \cite{Cheu21} that assesses the impact of poisoning attacks on tabular data.  Clearly, smaller the value of $\epsilon$, worse is the attack's impact.   

Similar to the Laplace mechanism, \DegRRNaive{} does not attain $\alpha$ malicious error for any $\alpha < n-1, \delta < 1$. This happens because of the strong attack where $\DO_n$ is an isolated node who acts maliciously and reports an all-one list. Thus, once again this  worst-case soundness is inevitable even with a single malicious user. 

\section{Corresponding Results for Input Poisoning}\label{app:input}
  We first investigate the baseline protocols \RLap{} and \DegRRNaive{} and show that while input poisoning attacks are less damaging than response poisoning attacks, the protocols still suffer from poor robustness guarantees. Next, we show that our proposed protocols, \DegRRCheck{} and \DegHybrid{}, offer improved robustness against input poisoning attacks. These results demonstrate a separation between the efficacy of response and input poisoning attacks.

Recall in the Laplace mechanism, each user simply reports a private estimate of their degree. Under input poisoning attacks, Laplace noise is added to the poisoned input before it is reported to the data aggregator. Consequently, the response poisoning attack in which a malicious user could \textit{deterministically} report their degree as $n-1$ (Thm.~\ref{thm:response:laplace}) is no longer possible -- in order to manipulate their degree by $n-1$, the malicious user needs to get lucky with the sampled Laplace noise, resulting in the following theorem:

 \begin{thm} Under input poisoning, for any $\delta > 0$, the \RLap~protocol has $\frac{1}{\epsilon}\ln\frac{n}{\delta}$ honest error. For any $\delta > \frac{1}{2}$, \RLap has $n-1$ malicious error. \label{thm:input:laplace}  
 \end{thm}
The proof of the above theorem is in \ifpaper the full paper \cite{Full}\else App.~\ref{app:thm:input:laplace}\fi.
Unsurprisingly, compared to Thm.~\ref{thm:response:laplace} for response poisoning, the honest error is unchanged because no attack is possible for honest users. However, the malicious error is different. Thm.~\ref{thm:response:laplace} delineates a strong attack which always results in $(n-1)$ malicious error. In contrast, \RLap{} has $n-1$ malicious error with $\delta = \frac{1}{2}$ with respect to any input poisoning attack. This is because the sampled Laplace noise is negative with probability $\frac{1}{2}$.  Hence, even a worst-case malicious user who sends the maximum degree of $n-1$ will only get assigned a final estimate this high if the sampled noise is non-negative. Thus, the noise in the Laplace mechanism prevents the adversary from carrying out the deterministic worst-case attack.

Next, we show the result for $\DegRRNaive{}$, our second baseline protocol. There is an improvement in both honest and malicious errors, because the adversary's signals in the poisoned data (such as, a malicious user indicating they share an edge with every other user, or $m$ malicious users intentionally deleting their edges to an honest user), are noised via randomized response which weakens them.

\begin{thm}\label{thm:b2a2_easy}  The \DegRRNaive{} protocol achieves $m+\sqrt{n}\frac{\sqrt{2(e^\epsilon+1) \ln \frac{4n}{\delta}}}{e^\epsilon-1}$ honest error against input poisoning, and $(n-1,\frac{1}{2})$ malicious error w.r.t any input poisoning from $\calM$.\label{thm:input:naive}
\end{thm}
The above theorem is proved in \ifpaper the full paper \cite{Full}\else App.~\ref{app:thm:b2a2_easy}\fi.
Written asymptotically, the honest error  of Thm.~\ref{thm:input:naive} is $\tilde{O}(m + \frac{\sqrt{n}}{\epsilon})$, which improves the guarantee over response poisoning attacks (Thm.~\ref{thm:response:naive}) by a factor of $\frac{m}{\epsilon}$. This shows  a separation between input and response poisoning attacks.  A similar case holds for soundness -- while \DegRRNaive{} is $(n-1)$-tight sound under response poisoning attacks, for input poisoning attacks, it is $(n-1, \frac{1}{2})$-sound. The implications of this observation are similar to those of Thm.~\ref{thm:input:laplace} as discussed. 



Despite exhibiting improvement over response poisoning attacks, both naive protocols still fall short of providing acceptable soundness guarantees.

\begin{table*}[tbh!]
\centering

\scalebox{0.74}{
\rotatebox{90}{
\begin{tabular}{|cccccc|}
\toprule 
\multirow{4}{*}{\textbf{Attack}} & \textbf{Number} & \textbf{Number} & \textbf{Number} & \textbf{Malicious} & \multirow{4}{*}{\textbf{Description}}\\
& \textbf{of} & \textbf{of} & \textbf{of} & \textbf{User} & \\
 & \textbf{Malicious}  & \textbf{Malicious}  & \textbf{Honest}  & \textbf{Selection}  & \\
 & \textbf{Non-targets} & \textbf{Targets} & \textbf{Targets} & \textbf{Strategy} & 
\\
\midrule
\multirow{2}{*}{A1} & \multirow{2}{*}{39} & \multirow{2}{*}{1} & \multirow{2}{*}{0} & \multirow{2}{*}{Random} & Pure degree inflation attack where the malicious  \\
& & & & & users are chosen at random from the entire graph \\\hline
\multirow{2}{*}{A2} & \multirow{2}{*}{40} & \multirow{2}{*}{0} & \multirow{2}{*}{1} & \multirow{2}{*}{Random} & Pure degree deflation attack where the malicious  \\
& & & & & users are chosen at random from the entire graph \\\hline
\multirow{2}{*}{A3} & \multirow{2}{*}{40} & \multirow{2}{*}{0} & \multirow{2}{*}{1} & \multirow{2}{*}{Neighbor} & Pure degree deflation attack where the malicious  \\
& & & & & users are neighbors of the honest target  \\\hline
\multirow{2}{*}{A4} & \multirow{2}{*}{35} & \multirow{2}{*}{5} & \multirow{2}{*}{0} & \multirow{2}{*}{Random} & Pure degree inflation attack where the malicious  \\
& & & & & users are chosen at random from the entire graph \\\hline
\multirow{2}{*}{A5} & \multirow{2}{*}{30} & \multirow{2}{*}{10} & \multirow{2}{*}{0} & \multirow{2}{*}{Random} & Pure degree inflation attack where the malicious  \\
& & & & & users are chosen at random from the entire graph \\\hline
\multirow{2}{*}{A6} & \multirow{2}{*}{40} & \multirow{2}{*}{0} & \multirow{2}{*}{5} & \multirow{2}{*}{Community} & Pure degree deflation attack where the malicious users and the \\
& & & & &  honest targets are chosen from the same community of the graph \\\hline
\multirow{2}{*}{A7} & \multirow{2}{*}{40} & \multirow{2}{*}{0} & \multirow{2}{*}{10} & \multirow{2}{*}{Community} & Pure degree deflation attack where the malicious users and  \\
& & & & & honest targets are chosen from the same community of the graph \\\hline
\multirow{2}{*}{A8} & \multirow{2}{*}{40} & \multirow{2}{*}{0} & \multirow{2}{*}{600} & \multirow{2}{*}{Community} & Pure degree deflation attack where the malicious users and the  \\
& & & & & honest targets are chosen from the same community of the graph \\\hline
\multirow{2}{*}{A9} & \multirow{2}{*}{35} & \multirow{2}{*}{5} & \multirow{2}{*}{5} & \multirow{2}{*}{Community} & Combination attacks where the malicious users and the   \\
& & & & & honest targets are chosen from the same community of the graph \\\hline
\multirow{2}{*}{A10} & \multirow{2}{*}{30} & \multirow{2}{*}{10} & \multirow{2}{*}{10} & \multirow{2}{*}{Community} & Combination attacks where the malicious users and the   \\
& & & & & honest targets are chosen from the same community of the graph \\\hline
\multirow{2}{*}{A11} & \multirow{2}{*}{(15,15)} & \multirow{2}{*}{(5,5)} & \multirow{2}{*}{(0,0)} & \multirow{2}{*}{(Community,Community)} & Degree inflation attack where two sets of 15 non-target malicious users target  \\
& & & & &   5 malicious users independently in two different communities  \\\hline
\multirow{2}{*}{A12} & \multirow{2}{*}{(10,10)} & \multirow{2}{*}{(10,10)} & \multirow{2}{*}{(0,0)} & \multirow{2}{*}{(Community,Community)} & Degree inflation attack where where two sets of 10 non-target malicious users target \\
& & & & &   10 malicious users independently in two different communities\\\hline
\multirow{2}{*}{A13} & \multirow{2}{*}{(20,20)} & \multirow{2}{*}{(0,0)} & \multirow{2}{*}{(5,5)} & \multirow{2}{*}{(Community,Community)} & Degree deflation attack where two sets of 20 malicious users target   \\
& & & & & 5 honest users independently in two different communities\\\hline
\multirow{2}{*}{A14} & \multirow{2}{*}{(20,20)} & \multirow{2}{*}{(0,0)} & \multirow{2}{*}{(10,10)} & \multirow{2}{*}{(Community,Community)} & Degree deflation attack where two sets of 20 malicious users target \\
& & & & & 10 honest targets independently in two different communities\\\hline
\multirow{2}{*}{A15} & \multirow{2}{*}{(15,20)} & \multirow{2}{*}{(5,0)} & \multirow{2}{*}{(0,5)} & \multirow{2}{*}{(Community,Community)} & Combination attacks where two sets of malicious users of sizes 15 and 20 target   \\
& & & & & 10 malicious and 10 honest targets  independently in two different communities\\\hline
\multirow{2}{*}{A16} & \multirow{2}{*}{(10,20)} & \multirow{2}{*}{(10,0)} & \multirow{2}{*}{(0,10)} & \multirow{2}{*}{(Community,Community)} & Combination attacks where two sets of malicious users of sizes 10 and 20 target  \\
& & & & & 10 malicious and 10 honest targets  independently in two different communities \\
 \bottomrule
\end{tabular}
}}
\caption{Summary of evaluated attacks}\label{tab:attacks}
 \vspace{-0.5cm}
\end{table*}
\section{Evaluation Cntd.}\label{app:attacks}
In this section, we describe the specific implementations of the attacks we use for our evaluation in Section \ref{sec:eval}.

Recall an attack consists of $m$ malicious users, where $m$ is known beforehand. Each malicious user may perform any of the following three actions: 1) lie about their own connections to changer their estimate, 2) target some subset of malicious users to change those estimates, and 3) target some subset of honest users to tamper with those estimates. We consider sixteen possible ways to do this in ways that would often occur in the real world. The methods appear in Table~\ref{tab:attacks}, and we describe them now.

In the simplest attacks (A1 - A3), we consider a set of compromised malicious users whose goal is to either inflate a target malicious user or deflate a target honest user. In A1 and A2, the compromised malicious users constitute a random subset of users. In A3, the compromised malicious users come from the neighbors of the target honest user, representing a worse attack. In these attacks, and for the rest, we assume the target malicious user lies about his connections to try to increase his degree while the target honest user follows protocol.

In A4 and A5, we scale up A1 to more malicious targets, which could happen if a group of compromised users is used to target multiple accounts. As there is no way to select friends of many targets at once, we omit A3.

In A6 through A8, we consider a more realistic selection strategy: Instead of malicious users drawn completely at random, we consider they are drawn from a \emph{community} in the graph, which is more realistic as similar accounts are often targeted together. We then consider these accounts targeting multiple honest users, and in A8 we simulate what would happen when the malicious users are used recklessly and target a large fraction of the community.

In A9 and A10, we again consider malicious users in a community, but this time they target both malicious and honest users.

Finally, in the rest of the attacks, we consider combinations of two of the previous attacks carried out independently. This is more realistic, as in the real world, different parties of malicious users may collude independently.

\ifpaper
\else
\subsection{Attacks Against \DegRRCheck{}}\label{app:attack-algos}

\subsubsection{Degree Inflation Attacks}
Let $\DO_t, t \in \calM$ denote the target malicious user. 
\\
\noindent\textbf{Input Poisoning.} In this attack, the non-target malicious users set the bit for $\DO_t$ to be $1$. The target malicious user constructs his input by setting $1$ for all other malicious users. They also report $1$ for honest users to which they share an edge. 

For honest users to which $\DO_t$ does not share an edge, $\DO_t$ flips some of the bits to $1$ with the hopes of artificially increasing his degree. He does this for a $r_1$-fraction of these neighbors. See Algorithm~\ref{alg:att-input-inf} for the details; we term this attack $A_{\DegRRCheck}^{inp}$. Note that if $r_1 = 0$, then the malicious user is being completely honest for these users and will not inflate his degree, and if $r_1 = 1$, then he lies about each of these users and will likely be caught. Thus, his strategy is to pick a value in between $0$ and $1$, and in the experiments we found that $r_1 = 15\%$ was a good tradeoff point.

\begin{algorithm}[bt]
  \caption{$A_{\DegRRCheck}^{inp}: \{0,1\}^n\mapsto\{0,1\}^n$ }\label{alg:att-input-inf}
  \Parameter{$\epsilon$ - Privacy parameter}

  \KwData{$l \in \{0,1\}^n$ - True adjacency list, $t$ - Target honest user}
  \KwResult{$q \in \{0,1\}^n$ - Reported adjacency list}
  Select $r_1\in [0,1]$\; 
  \Comment{$\calH_1$ is the set of honest users with a mutual edge}
  $\calH_1=\{i \in \calH| l[i]=1\}$\;
  \Comment{$\calH_0$ is the set of honest users without a mutual edge}
  $\calH_0=\calH\setminus\calH_1$\;
  \Comment{Randomly sample $r_1$ fraction of the users in $\calH_0$}
  $F\in_R \calH_0, |F|=r_1|\calH_0|$\;
  $l'=\{0,0,\cdots, 0\}$\;
  \lFor{ $i \in \calH_1\cup\calM \cup F$}{ 
  $l'[i]=1$
  }
  \lFor{$i \in [n]$}{
    $q[i]=\rr_\rho(l'[i])$
  }
  \KwRet $q$
\end{algorithm}

\noindent \textbf{Response Poisoning.}  For response poisoning, the non-target malicious first find a plausible response by applying $RR_\rho$ to their data. They then set the bit for $\DO_t$ to be $1$, indicating they are connected to this user.

The target malicious user constructs his response by first applying $RR_\rho$ to his data to compute a plausible response. Then, he flips his bits to malicious users to $1$, and for honest users, he takes a $r_1$-fraction of the $0$s in his response and flips them to $1$. The quantity $r_1$ is a tradeoff parameter with the same intuition as for $A_{\DegRRCheck}^{inp}$. The details of this attack appear in Algorithm~\ref{alg:att-resp-inf}, and it is termed $A_{\DegRRCheck}^{resp}$.

\begin{algorithm}[bt]
  \caption{$A_{\DegRRCheck}^{resp}: \{0,1\}^n\mapsto\{0,1\}^n$ }\label{alg:att-resp-inf}
  \Parameter{$\epsilon$ - Privacy parameter}
  
  \KwData{$l \in \{0,1\}^n$ - True adjacency list}
  \KwResult{$q \in \{0,1\}^n$ - Reported adjacency list}
  
  Select $r_1\in [0,1]$\;
  $q = \textsf{RR}_\rho (l)$ \;
  \Comment{$\calI_1$ is the set of honest users with an edge in $q$.}
$\calI_1=\{i \in \calH| q[i]=1\}$\; 
$\calI_0 = \calH \setminus \calI_1$\;
\Comment{ Randomly sample $r_1$ fraction of the users in $\calI_0$.}
$F\in_R \calI_0, |F|=r_1|\calI_0|$\;
  \lFor{ $i \in \calI_1 \cup\calM \cup F$}{$q[i]=1$}
  \KwRet $q$
\end{algorithm}

\subsubsection{Degree Deflation Attacks}
Let $\DO_t, t \in \calH$ denote the target honest user. \\\noindent\textbf{Input Poisoning.} Here, every malicious user constructs his input acting honestly for non-target users and setting a $0$ for $\DO_t$.
\\\noindent\textbf{Response Poisoning.} 
Every malicious user acts honestly for non-target users by applying randomized response to their input. They finally send a $0$ for their connection to $\DO_t$.
\subsection{Attacks Against \DegHybrid{}}

\subsubsection{Degree Inflation Attacks}
Let $\DO_t, t \in \calM$ be the target malicious user.
\\
\noindent\textbf{Input Poisoning.} The non-target malicious users flip their edge to $\DO_t$ to a $1$ as they do in $A_{\DegRRCheck{}}^{inp}$. They send an honest estimate of their degree $\tilde{d}^{Lap}$ as this does not affect the target.

The target malicious user crafts his input adjacency list $q$ as he did in $A_{\DegRRCheck}^{inp}$. For his estimate $\tilde{d}_t^{Lap}$, he 
computes the expected value of $\tilde{d}_t^{rr}$ given that he submitted $q$ while the other users either submit $RR_\rho(l_i)$ or $RR_\rho(1)$, depending if they are honest or malicious. Specifically, the expected value is given by
\[
    e^{rr,inp} = \frac{m(1-\rho)^2 + \E[\sum_{i \in \calH} q_i RR_\rho(l_i)] - \rho^2 n}{1-2\rho}.
\]
He finally sets $\tilde{d}_t^{rr} = e_t^{rr,inp} + r_2 \frac{\tau}{1-2\rho}$ where $r_2 \in [0,1]$, which again trades off between how much cheating is possible and getting flagged. During the trials, we used $q_2 = 0.1$ as this did not significantly increase the target's chance of being rejected as $\bottom$. This attack, termed $A_{\DegHybrid}^{inp}$, appears in Algorithm~\ref{alg:att-resp-inf2}.

\begin{algorithm}[bt]
  \caption{$A_{\DegHybrid}^{inp}: \{0,1\}^n\mapsto\{0,1\}^n$ }\label{alg:att-resp-inf2}
  \Parameter{ $\epsilon$ - Privacy parameter}

  \KwData{ $l \in \{0,1\}^n$ - True adjacency list}
  \KwResult{ $q \in \{0,1\}^n$ - Reported adjacency list, $\tilde{d}^{lap}$ - Reported noisy degree estimate}
Select $r_2 \in [0,1]$\;
\Comment{$c$ is the constant used in Alg. \ref{alg:deghybrid} to divide the budget between the RR and Laplace steps.}
$\rho = \frac{1}{1+e^{c\epsilon}}$\;
$q \gets A_{\DegRRCheck}^{inp}(l, c\epsilon)$\;
$\tilde{count}^{11} \gets m(1-\rho)^2 + \E[\sum_{i \in \calH} q_i RR_\rho(l_i)]$\;
$ e^{rr,inp} = \frac{\tilde{count}^{11} - \rho^2 n}{1-2\rho}$\;
$ \hat{d}^{Lap} = e^{rr,inp} + r_2 \frac{\tau}{1-2\rho} + \eta$ where $\eta \sim Lap(\frac{1}{(1-c)\epsilon})$\;
\KwResult{ $q,\tilde{d}^{Lap}$}
\end{algorithm}

\noindent\textbf{Response Poisoning.}
The non-target malicious users flip their edge to $\DO_t$ to a $1$ as they do in $A_{\DegRRCheck{}}^{inp}$. They send an honest estimate of their degree $\tilde{d}^{Lap}$ as this does not affect the target.

The target malicious user crafts his response adjacency list $q$ as he did in $A_{\DegRRCheck}^{resp}$. For his estimate $\tilde{d}_t^{Lap}$, he computes the expected value of $\tilde{d}_t^{rr}$ given that he submitted $q$ while the other users either submit $RR_\rho(l_i)$ or $1$, depending if they are honest or malicious. This expected value is given by
\[
e^{rr, resp} = \frac{m + \E[\sum_{i \in \calH} q_i RR_\rho(l_i) - \rho^2 n]}{1-2\rho}.
\]
He finally sets $\tilde{d}_t^{rr} = e^{rr,resp} + r_2 \frac{\tau}{1-2\rho}$ where $r_2 \in [0,1]$ serves a similar tradeoff purpose as for $A_{\DegHybrid}^{inp}$. 

\begin{algorithm}[bt]
  \caption{$A_{\DegHybrid}^{resp}: \{0,1\}^n\mapsto\{0,1\}^n$ }
  \Parameter{$\epsilon$ - Privacy parameter}
  
  \KwData{ $l \in \{0,1\}^n$ - True adjacency list}
  \KwResult{$q \in \{0,1\}^n$ - Reported adjacency list, $\tilde{d}^{Lap}$ - Reported noisy degree estimate}
  
  Select $r_2\in [0,1]$ \;
  $q=A_{\DegRRCheck}^{resp}(l, \epsilon)$\;
$\rho=\frac{1}{1+e^{c\epsilon}}$\;
\Comment{$c$ determines how the privacy budget is divided between the two types of response as in Alg. \ref{alg:deghybrid}}
$\tilde{count}^{11} \gets m + \E[\sum_{i \in \calH} q_i RR_\rho(l_i)]$\;
$e^{rr,resp} = \frac{m + \tilde{count}^{11} - \rho^2 n}{1-2\rho}$\;
$\tilde{d}^{Lap} = e^{rr,resp} + r_2 \frac{\tau}{1-2\rho}$\;
\KwRet{ $q,\tilde{d}^{Lap}$}
\end{algorithm}

\subsubsection{Degree Deflation Attacks}
Let $\DO_t, t \in \calH$ represent the honest target.\\
\noindent\textbf{Input Poisoning.}
For the adjacency list, all the malicious users follow the same protocol as for \DegRRCheck{}. For the degree, all the malicious users follow the Laplace mechanism truthfully as these values are immaterial for estimating the degree of the target honest user.
  
\noindent\textbf{Response Poisoning.} 
For the adjacency list, all the malicious users follow the same protocol as for \DegRRCheck$(\cdot)$. For the degree, all the malicious users follow the Laplace mechanism truthfully as these values are immaterial for estimating the degree of the target honest user.

\section{Proofs of Lower Bounds}\label{app:lb-proofs}
We will begin with definitions and preliminary results, and then prove our two lower bounds for any protocol.
\subsection{Definitions and Preliminary Results}
Our results will appeal to information theory, and will specifically need the $KL$ divergence $D_{KL}$ between two distributions. It is sufficient for us to consider probability distributions over discrete domains; our results transfer easily to continuous domains (using the relevant results from measure theory).
\begin{defn}
    For two distributions $P, Q$ on a space discrete space $\calX$, the $KL$ divergence is defined as
    \[
        D_{KL}(P \| Q) = \sum_{x \in \calX} P_x \ln (\tfrac{P_x}{Q_x}).
    \]
\end{defn}
A related, simpler measure will the the total variation distance (TVD):
\begin{defn}
    For two distributions $P, Q$ on a space discrete space $\calX$, the TVD between $P,Q$ is defined by
    \[
        TVD(P,Q) = \frac{1}{2} \sum_{x \in \calX} |P_x - Q_x|.
    \]
\end{defn}
These two quantities are both $f$-divergences, which mean they satisfy the data processing inequality:
\begin{fact}
    For two distributions $P,Q$ on $\calX$, and any (possibly random) function $f : \calX \rightarrow \calY$, we have
    \[
        D_{KL}(f(P) \| f(Q)) \leq D_{KL}(P \| Q).
    \]
    The same inequality holds for TVD.
\end{fact}
For a proof of this and many other relevant results, see~\cite{cover1999elements}.

It is important for us to have the following approximation to $KL$ divergence:
\begin{lemma}\label{lem:kl-approx}
    For $p, q \in [0,1]$ and $P = Bern(p), Q = Bern(q)$, we have $D_{KL}(P\|Q) \leq \frac{2(p-q)^2}{\min\{q, 1-q\}}$.
\end{lemma}
\begin{proof}
    By the definition of KL divergence, we have
    \begin{align*}
        D_{KL}(P \| Q) &= \ln(\tfrac{p}{q})p + \ln(\tfrac{1-p}{1-q})(1-p) \\
        &= q(\tfrac{p}{q} \ln (\tfrac{p}{q}) - \tfrac{p}{q} + 1) + (1-q)(\tfrac{1-p}{1-q} \ln (\tfrac{1-p}{1-q}) - \tfrac{1-p}{1-q} + 1).
    \end{align*}
    We have for all $x$ that $\ln(x) - x + 1 \leq (x-1)^2$ using Taylor expansion. Thus,
    \begin{align*}
    D_{KL}(P \| Q) &\leq q (\tfrac{p}{q} - 1)^2 + (1-q)(\tfrac{1-p}{1-q} - 1)^2 \\
    &\leq \frac{(p-q)^2}{q} + \frac{(p-q)^2}{1-q},
    \end{align*}
    and the result follows.
\end{proof}
Next, we will need the following, very general result from the differential privacy literature~\cite{kairouz2015composition}, which will allow to simplify our later analysis by considering just randomized response.
\begin{lemma}\label{lem:ldp-postproc}
    (From~\cite{kairouz2015composition})
    If $\calR(0), \calR(1)$ satisfy $\epsilon$-LDP, then there exist distributions $R^0, R^1$ such that
    \begin{align*}
        \calR(0) &= (1-q) R^0 + q R^1 \\
        \calR(1) &= q R^0 + (1-q) R^1,
    \end{align*}
    where $q = \frac{1}{e^\epsilon + 1}$. (Note we are slightly abusing notation, as the sum of distributions $q P + (1-q)Q$ means the distribution whose p.d.f. is the sum of the p.d.f's of $P,Q$).
\end{lemma}
In particular, Lemma~\ref{lem:ldp-postproc} implies we may view any $\epsilon$-LDP protocol with input $x \in \{0,1\}$ as a post-processing function $\calP$ applied to $RR_q(x)$, where $RR_q$ denotes randomized response with flipping parameter $q$. 

Finally, we need a result which upper bounds the maximum amount of information, measured by the KL divergence, which an LDP protocol is able to release. There is a long line of work on this problem starting with~\cite{duchi2013local}. We only need the following particular case, which we prove here for completeness:
\begin{lemma}\label{lem:kl-ub}
    If $\calR(x)$ satisfies $\epsilon$-LDP, and $x^0, x^1$ are drawn from two distributions $P,Q$ on $\{0,1\}$, then
    \[
    D_{KL}(\calR(x^0) \| \calR(x^1)) \leq \frac{(e^\epsilon - 1)^2}{e^\epsilon+1} TVD(P,Q)^2.
    \]
\end{lemma}
\begin{proof}
    We apply Lemma~\ref{lem:ldp-postproc} to $\calR$, meaning we can write $\calR(x) = \calP(RR_q(x))$ with $q = \frac{1}{e^\epsilon + 1}$. Using the post-processing inequality of $KL$ divergence, 
    \[
        D_{KL}(\calR(x^0) \| \calR(x^1)) \leq D_{KL}(RR_q(x^0) \| RR_q(x^1)).
    \]
    Define $P = Bern(a)$ and $Q = Bern(b)$. We have that $RR_q(x^0) \sim Bern(a')$, where $a' = a(1-q) + (1-a)q = \frac{ae^\epsilon - a + 1}{e^\epsilon+1}$, and $RR_q(x^1) \sim Bern(b')$ with $b' = \frac{b e^\epsilon - b + 1}{e^\epsilon+1}$. Applying Lemma~\ref{lem:kl-approx}, we have 
    \begin{align*}
        D_{KL}(P \| Q) &= \frac{(a'-b')^2}{\min \{b', 1-b'\}} \\
        &= \frac{(a-b)^2 \left(\frac{e^\epsilon-1}{e^\epsilon+1}\right)^2}{\frac{1}{e^\epsilon+1} + \min\{b, 1-b\}\left(\frac{e^\epsilon-1}{e^\epsilon+1}\right)} \\
        &\leq \frac{(a-b)^2(e^\epsilon-1)^2}{e^\epsilon+1},
    \end{align*}
    as desired.
\end{proof}

\subsection{Proof of Theorem~\ref{thm:input-lb}}

Suppose to the contrary that there was such a protocol with $\alpha$-honest and malicious error with $\alpha = \frac{m}{40} + \frac{\sqrt{n(e^\epsilon+1)}}{80(e^\epsilon-1)}$ and $\delta = 0.1$. Let $\{\calR_i\}_{i=1}^n$ be the local randomizers of the protocol. Consider an ``honest'' world where $G$ has no edges except to $\DO_n$. For $i = 1$ to $n-1$, let $z_i$ indicate whether the edge from $i$ to $n$ is present, and let $z_i \sim Bern(\frac{1}{2} + p)$ i.i.d. where $p =\frac{10\alpha}{n}$. Let $y_i$ indicate whether $\DO_i$ acts maliciously, and let $y_i \sim Bern(\frac{m}{2n})$, i.i.d. Finally, suppose that each malicious with an edge to $\DO_n$ behaves as if the edge doesn't exist. Formally, for $i = 1$ to $n-1$, the random variable 
\[ x_i^0 = 
\begin{cases}
    0 & y_i = 1 \\
    z_i & \text{otherwise}
\end{cases}
\]
defines the input of $\DO_i$ to the randomizer that concerns edges to $\DO_n$. Next,
$\DO_n$, who will behave honestly in this world, will use $w_i^0 = z_i$ for $i = 1, \ldots, n-1$, as their input. 
Thus, the responses in this world will be drawn from
\[
    \left( \calR_1(\langle \textbf{0}, x_1^0\rangle), \ldots, \calR_{n-1}(\langle \textbf{0}, x_{n-1}^0 \rangle), \calR_n(\langle w_1^0, \ldots, w_{n-1}^0 \rangle ) \right),
\]
where $\textbf{0}$ indicates a vector of $n-1$ 0s. Let $r^{0}_{(n-1)}$ indicate the first $n-1$ terms of the above tuple, and $r^{0}_n$ indicate the final term. By the honest error guarantee, the responses may be post-processed into an estimate $\hd_n(r^{0}_{(n-1)}, r^0_n)$ such that $|\hd_n - d_n| \leq \alpha$ for each graph $G$. Using Hoeffding's inequality, with probability at least $0.9$, $\DO_n$ has degree at least $\frac{n}{2} + 10 \alpha - 2 \sqrt{n}$ which is at least $\frac{n}{2} + \alpha$ since $\epsilon < 0.5$. Thus, $\Pr[\hd_n \geq \frac{n}{2}] \geq 0.8$, where the probability considers the randomness in $G$, and running the protocol described above.

In the second, ``malicious'' world, $G$ will be defined the same way, except each $z_i \sim Bern(\frac{1}{2} - p)$ i.i.d. For $i = 1$ to $n-1$, malicious users will be chosen the same way as before, and we may define $y_i$ as before. This time, each malicious user will behave as if there is an edge to user $n$, so they will use 
\[ x_i^1 = 
\begin{cases}
    1 & y_i = 1 \\
    z_i & \text{otherwise}
\end{cases}
\]
as their input to $\DO_n$. $\DO_n$ will act maliciously as well, by acting as if he has an edge to each malicious user, and with probability $s = 1-\frac{(1-2p)}{(1+2p)(1-m/2n)}$, to an honest user as well. Formally, he will define $y_i' \sim Bern(s)$, and compute
\[ w_i^1 = 
\begin{cases}
    1 & y_i = 1 \text{ or } y_i' = 1 \\
    z_i & \text{otherwise}
\end{cases}
\]
Thus, the output of the protocol in this world will be 
\[
    \left( \calR_1(\langle \textbf{0}, x_1^1\rangle), \ldots, \calR_{n-1}(\langle \textbf{0}, x_{n-1}^1 \rangle), \calR_n(\langle w_1^1, \ldots, w_{n-1}^1 \rangle ) \right).
\]
Let $r^{1}_{(n-1)}$ indicate the first $n-1$ terms of the above tuple, and $r^{1}_n$ indicate the final term. By the error guarantees, we have that in this world, $\Pr[\hd_n = \bot \vee \hd_n \leq \frac{n}{2}] \geq 0.8$. This is a disjoint event from the event in the honest world, and in particular, it implies that $TVD((r_{(n-1)}^0, r_n^0), (r_{(n-1)}^1, r_n^1)) \geq 0.6$.

However, observe that each $w_i^0$ is identically distributed to each $w_i^1$---they are both drawn from $Bern(\frac{m}{2n} + (1-\frac{m}{2n})(\frac{1}{2}+p))$. In the honest world, we have that $\Pr[x_i^0 = 0 | w_i^0 = 0] = 1$, and $\Pr[x_i^0 = 0 | w_i^0 = 1] = \Pr[y_i = 1]$ since the only way $x_i^0 = 0$ can occur if $z_i^0 = 1$ is if $y_i = 1$.

Similarly, in the malicious world we have $\Pr[x_i^1 = 0 | w_i^1 = 0] = 1$, and
\begin{align*}
    &\Pr[x_i^1 = 0 | w_i^1 = 1] \\
    &= \frac{\Pr[x_i^1 = 0, w_i^1 = 1]}{\Pr[w_i^1 = 1]} \\
    &= \frac{\Pr[y_i = 0, y_i' = 1, z_i = 0]}{1-\Pr[y_i = 0, y_i' = 0, z_i = 0]} \\
    &= \frac{2p - \frac{m}{4n} + p\frac{m}{2n}}{1-(1-m/2n)(1-s)(1/2+p)} \\
    &= \frac{2p - \frac{m}{4n} + p\frac{m}{2n}}{1/2+p}
\end{align*}

Now, we will derive a contradiction using the information between the two worlds. In either world,we have that $r_n$ is a post-processing of $w_{(n-1)}$. By the data processing inequality and chain rule of KL divergences, we have
\begin{align*}
    &D_{KL}((r_{(n-1)}^0, r_n^0) \| (r_{(n-1)}^1, r_n^1)) \\
    &\leq D_{KL}((r_{(n-1)}^0, w_{(n-1)}^0) \| (r_{(n-1)}^1, w_{(n-1)}^1)) \\ &= D_{KL}(w_{(n-1)}^0 \| w_{(n-1)}^1) + D_{KL}(r_{(n-1)}^0 | w_{(n-1)}^0 \| r_{(n-1)}^1 | w_{(n-1)}^1).
\end{align*}
The first term is clearly $0$. Using conditional independence of the $x_i$s and the protocols, we can write 
\begin{align*}
    &D_{KL}(r_{(n-1)}^0 | w_{(n-1)}^0 \| r_{(n-1)}^1 | w_{(n-1)}^1) \\
    &= \sum_{i=1}^{n-1} D_{KL}(r_{i}^0 | w_{i}^0 \| r_{i}^1 | w_{i}^1).
\end{align*}
Now, we apply Lemma~\ref{lem:kl-ub}, which states that
\[
    D_{KL}(r_i^0 | w_i^0 \| r_i^1 | w_i^1) \leq \frac{(e^\epsilon-1)^2}{e^\epsilon+1}TVD(r_i^0 | w_i^0, r_i^1 | w_i^1).
\]
We can solve
\begin{align*}
    TVD(x_i^0 | w_i^0 = 0, x_i^1 | w_i^1 = 0) &= 0 \\
    TVD(x_i^0 | w_i^0 = 1, x_i^1 | w_i^1 = 1) &= \left|\frac{2p - \frac{m}{4n} + p\frac{m}{2n}}{1/2+p} - \frac{m}{2n}\right| \\
    &= \frac{|2p - m/2n|}{1/2+p}.
\end{align*}
Plugging in $p = \frac{10 \alpha}{n} = \frac{m}{4n} + \frac{\sqrt{e^\epsilon+1}}{8\sqrt{n}(e^\epsilon-1)}$, we may bound the second equation as $\frac{\sqrt{e^\epsilon+1}}{2\sqrt{n}(e^\epsilon-1)}$ and thus $D_{KL}(r_i^0 | w_i^0 \| r_i^1 | w_i^1) \leq \frac{1}{4n} $. Thus, we have
\[
D_{KL}(r_{(n-1)}^0 | w_{(n-1)}^0 \| r_{(n-1)}^1 | w_{(n-1)}^1) \leq \frac{1}{4}.
\]

Finally, by Pinsker's inequality, we know that 
\begin{align*}
    TVD((r_{(n-1)}^0, r_n^0), (r_{(n-1)}^1, r_n^1)) \leq \sqrt{\frac{1}{32}} < 0.5, 
\end{align*}
completing the contradiction.

\subsection{Proof of Theorem~\ref{thm:output-lb}}

    Suppose to the contrary there is a protocol given by local randomizers $\calR_i$ which attains both $\alpha$ honest and malicious error for $\delta = 0.1$, and $\alpha \approx \frac{m}{4 \epsilon} + \frac{\sqrt{n}}{40\epsilon}$, whose specific value we will set later.

    Consider an ``honest'' world where $G$ has no edges except to user $n$. For $i = 1$ to $n-1$, let $z_i^0$ be a variable indicating whether the edge from $i$ to $n$ is present and drawn from $ Bern(\frac{1}{2} + p)$ i.i.d. where $p = \frac{2\alpha}{n}$. Each user will act completely honestly: thus, for each $1 \leq i \leq n-1$, the server will receive the pair $(r_i^0, s_i^0)$ about the variable $z_i^0$, where $r_i^0$ is an $\epsilon$-LDP response from user $i$, and $s_i^0$ is an $\epsilon$-LDP response from user $n$ (because of the assumption that user $n$'s protocol is edgewise factorable). Thus, $(r_i^0, s_i^0)$ is the result of a $2\epsilon$-LDP protocol $\calR_i'$ applied to $z_i^0$. By Lemma~\ref{lem:ldp-postproc}, there exist distributions $R_i^0, R_i^1$ such that $(r_i^0, s_i^0)$ are drawn from the distribution 
    \begin{align*}
        &\;(\tfrac{1}{2}-p) \calR_i'(0) + (\tfrac{1}{2}+p) \calR_i'(1) \\
        &\qquad = (\tfrac{1}{2}-p) ((1-q') R_i^0 + q' R_i^1) + (\tfrac{1}{2}+p) (q' R_i^0 + (1-q') R_i^1) \\
        &\qquad = (\tfrac{1}{2}-p+2pq')R_i^0 + (\tfrac{1}{2}+p-2pq')R_i^1 \\
        &\qquad = a^0 R_i^0 + (1-a^0) R_i^1
    \end{align*}
    where $q' = \frac{1}{1 + e^{2\epsilon}}$.

    In the malicious world, setup will be similar. We will let $z_i^1$ indicate whether the edge from $i$ to $n$ is present, and it is drawn from $Bern(\frac{1}{2}-p)$ i.i.d. Users $n$ and some users from $1$ to $n-1$ will act maliciously. In particular, the variable $y_i \sim Bern(\frac{m}{2n})$ will indicate whether user $i$ acts maliciously. If $y_i = 1$, then users $i$ and $n$ will collude so that their responses are drawn from $R_i^1$. Otherwise, user $i$ and user $n$ will behave honestly. Therefore, the responses $(r_i^1, s_i^1)$ in this world are drawn from the distribution
    \begin{align*}
        &\;(1-\frac{m}{2n})\left((\tfrac{1}{2}+p) \calR_i'(0) + (\tfrac{1}{2}-p) \calR_i'(1)\right) + \tfrac{m}{2n} R_i^1 \\
        &\qquad=(1-\frac{m}{2n})\left((\tfrac{1}{2}+p-2pq')R_i^0 + (\tfrac{1}{2}-p+2pq')R_i^1\right) + \tfrac{m}{2n} R_i^1 \\
        &\qquad=(1-\frac{m}{2n})(\tfrac{1}{2}+p-2pq')R_i^0 + \left((1-\frac{m}{2n})(\tfrac{1}{2}-p+2pq') + \tfrac{m}{2n} \right)R_i^1 \\
        &\qquad= a^1 R_i^0 + (1-a^1)R_i^1.
    \end{align*}
    By Lemma~\ref{lem:kl-approx} and the post-processing inequality, we have 
    \[
        D_{KL} (r_i^0, s_i^0 \| r_i^1, s_i^1) \leq \frac{2(a^0 - a^1)^2}{\min\{a^1, 1-a^1\}}.
    \]
    
    We can solve 
    \[
    a^0 - a^1 = \frac{m}{4n} - p(2-\frac{m}{2n})(1-2q'),
    \]
    and by setting $p = \frac{m/4n}{(2-\frac{m}{8n})(1-2q')} + \frac{1}{2(1-2q')\sqrt{n}}$, we will have $(a^0 - a^1)^2 \leq \frac{1}{16n}$. This will result in a value of $a^1$ which is nearly $\frac{1}{2}$, and we obtain $D_{KL}(r_i^0, s_i^0 \| r_i^1, s_i^1) \leq \frac{1}{4n}$. This shows that
    \[
        D_{KL}(r_1^0, \ldots, r_n^0, s_1^0, \ldots, s_n^0 \| r_1^1, \ldots, r_n^1, s_1^1, \ldots, s_n^1) \leq \frac{1}{4}.
    \]
    Applying Pinsker's inequality, we have
    \[
    TVD(r_1^0, \ldots, r_n^0, s_i^0, \ldots, s_n^0 \| r_1^1, \ldots, r_n^1, s_i^1, \ldots, s_n^1) \leq \sqrt{\frac{D_{KL}}{2}} < 0.5.
    \]
    
    However, taking $\alpha \leq \frac{n p}{10}$, in the honest world, the true degree of user $n$ is at least $\frac{n}{2} + \alpha$ with probability at least $0.9$; in the malicious world, the true degree of user $n$ is at most $\frac{n}{2} - \alpha$ with probability at least $0.9$. Following a similar argument to that of Theorem~\ref{thm:input-lb}, we can show
    \[ TVD(r_1^0, \ldots, r_n^0, s_i^0, \ldots, s_n^0 \| r_1^1, \ldots, r_n^1, s_i^1, \ldots, s_n^1) \geq 0.6,\]

    violating our TVD lower bound.

\section{Proofs of Algorithmic Guarantees}\label{app:proofs}
First, we introduce notation and preliminary results used in our proofs.
\subsection{Notation} In this section, for a graph $G$ with vertices $[n]$, we let $d_i(S)$ for $S \subseteq [n]$ denote the number of neighbors of node $i$ in the set $S$.
We will often abuse notation for a set $\calS$ of users by also letting $\calS$ be the indices of the users in the set. Thus, we may let $i \in \calS$ be the index of some user in $\calS$.
Finally, we sometimes refer to user $\DO_i$ simply as user $i$.

\subsection{Preliminary Results}
We will heavily make use of the following concentration result:

\begin{lemma}\label{lem:bern-concentration}
    Let $X_1, \ldots, X_n$ denote independent random variables such that $X_i \sim \bern(p_i)$. Let $v = \sum_{i=1}^n p_i(1-p_i)$, and $X = \sum_{i=1}^n X_i$. Then,
    \begin{align*}
        \Pr[|X - \E[X]| \geq \max\{1.5 \ln \frac{2}{\delta}, \sqrt{2v\ln \frac{2}{\delta}}\}] &\leq \delta.
    \end{align*}
\end{lemma}
\begin{proof}
    Center the random variables so that $Z_i = X_i - p_i$; the variance $v$ does not change. We know by Bernstein's inequality that for all $t \geq 0$,
    \[
        \Pr[Z \geq t] \leq \exp\left(\frac{-t^2}{2(v + t / 3)}\right) \leq \exp\left(- \max \left\{ \frac{t^2}{2v}, \frac{3t}{2}\right\}\right).
    \]
    Thus, if $t \geq \max\{\frac{3}{2} \ln \frac{2}{\delta}, \sqrt{2 v \ln \frac{2}{\delta}}\}$, then $\Pr[Z \geq t] \leq \frac{\delta}{2}$. Applying the argument to $-Z$, we obtain the two-sized bound.
    
\end{proof}

Next, we observe the following facts about randomized response.
\begin{fact}\label{fact:rr-exp}
If user $i \in \calH$, then $\E[q_i[j]] = \rho + (1-2\rho) d_i(j)$.
\end{fact}

\begin{fact}\label{fact:2rr-exp}
If users $i,j \in \calH$, then $\E[q_i[j]q_j[i]] = \rho^2 + (1-2\rho) d_i(j)$.
\end{fact}
\subsection{Proof of Theorem \ref{thm:response:laplace}}\label{app:thm:response:laplace}
Recall that in the Laplace mechanism, a user's degree estimate $\hat{d}_i$ is simply $d_i + L_i$, where $L_i \sim Lap(\frac{1}{\epsilon})$ is a Laplace random variable generated by the user.\\\\\noindent
\textbf{Honest error.} The  guarantee for honest error follows from the concentration of Laplace distribution: Each Laplace random variable $L_i$ satisfies $|\Pr[|L_i| \geq t] \leq e^{-t\epsilon}$. Setting $t = \frac{1}{\epsilon}\ln \frac{n}{\delta}$ and applying the union bound, each of the $n$ Laplace variables will satisfy $|L_i| \leq \frac{1}{\epsilon}\ln \frac{n}{\delta}$ with probability $1-\delta$, and if this holds, then $|d_i - \hat{d}_i| \leq \frac{1}{\epsilon}\ln \frac{n}{\delta}$ for honest users.
\\
\noindent\textbf{Existence of Attack} Consider the empty graph. A malicious user $\DO_i$ may report $n-1$, the maximum possible degree, and thus $\hat{d}_i = n-1$ while $d_i = 0$.

\subsection{Proof of Theorem~\ref{thm:response:naive}}\label{app:thm:response:naive}
\noindent \textbf{Honest Error.}

As defined in \DegRRNaive{}, the estimator $count_i^1$ is given by 
\begin{gather}count^1_i=(\sum_{j < i} q_j[i] + \sum_{i < j} q_i[j])\end{gather}
We may alternatively split the above sum into honest bits and malicious bits as $count^1_i = hon_i + mal_i$. Here, 
\begin{gather*}
hon_i = \sum_{j < i, j \in \calH} q_j[i] + \sum_{i < j} q_i[j] \\
mal_i = \sum_{j < i, j \in \calM} q_j[i].
\end{gather*}
Since all bits in the sum $hon_i$ are honest, by Fact~\ref{fact:rr-exp} we have $\E[hon_i] = \rho|\calH_i| + (1-2\rho) d_i(\calH_i)$, where $\calH_i = \calH \cup \{1, 2, \ldots, i-1\}$. 

Furthermore, $0 \leq mal_i \leq |\calM_i|$, where $\calM_i = [n] \setminus \calH_i$. This implies $|mal_i - E_{mal,i}| \leq |\calM_i|$, where $E_{mal,i} = \rho|\calM_i| + (1-2\rho) d_i(\calM_i)$.
By Lemma~\ref{lem:bern-concentration} and a union bound, with probability $1-\delta$, we have for all $i \in \calH_i$ that
\begin{gather*}
\left|hon_i - \E[hon_i] + mal_i - E_{mal,i} \right| \leq  \sqrt{2\rho n \ln \frac{2n}{\delta}} + |\calM_i| \\ 
\implies \left|count_i^1 - \rho n - (1-2\rho)d_i \right| \leq  \sqrt{2\rho n \ln \frac{2n}{\delta}} + m \\ 
\implies |\hat{d}_i - d_i | \leq \frac{1}{1-2\rho} \sqrt{2\rho n \ln \frac{2n}{\delta}} + \frac{m}{1-2\rho}.
\end{gather*}

\noindent\textbf{Existence of the Attack} Consider the empty graph, and suppose that user $n$ is malicious. Since this user reports all his edges, he may report $q_i[j] = 1$ for all $j < 1$. Thus, $\hat{d}_n \geq n-1$, but $d_n = 0$, showing that $\Pr[\mathrm{err}^{mal}(\hat{d}_n, d_n) = n-1] = 1$.

\subsection{Proof of Theorem~\ref{thm:response:check}} \label{app:b3a3}
Recall the key quantities defined in \DegRRCheck{} (Algorithm~\ref{alg:degrrcheck}):
\begin{gather}
      count_i^{11} = \sum_{j \in [n] \setminus i} q_{i}[j] q_{j}[i] \\
      count_i^{01} = \sum_{j \in [n] \setminus i} (1-q_{i}[j])q_{j}[i].
\end{gather}
We now prove honest error.

\noindent \textbf{Honest error.} 
It will be helpful to split $count_i^{11} = hon_i^{11} + mal_i^{11}$, where $hon_i^{11} = \sum_{j \in \calH \setminus i} q_{i}[j] q_{j}[i]$ and $mal_i^{11} = \sum_{j \in \calM \setminus i} q_{ij} q_{ji}$. We define $hon_i^{01}$ and $mal_i^{01}$ similarly such that they satisfy $count_i^{01} = hon_i^{01} + mal_i^{01}$. We break the proof into two claims: showing that honest users receive an accurate estimate and that they are not disqualified.

\begin{claim}\label{claim:honest-response-concentration-1} We have
\[
    \Pr[\forall \DO_i \in \calH.~|\hat{d}_i - d_i| \geq \tfrac{m + 2 \sqrt{\rho n \ln \frac{4n}{\delta}}}{1-2\rho}] \leq \frac{\delta}{2}.
\]
\end{claim}
\begin{proof} 
Let $\DO_i \in \calH$. Then, $hon_i^{11}$ is a sum of $h-1$ Bernoulli random variables with $p = \rho^2$ or $(1-\rho)^2$. By Fact~\ref{fact:2rr-exp}, we have
\begin{align*}
    \E[hon_i^{11}] &= \rho^2 (h-1) + (1-2\rho) d_i(\calH)
\end{align*}
Now, $v$ defined in Lemma~\ref{lem:bern-concentration} satisfies $(h-1) \rho^2 \leq v \leq (h-1)(1-(1-\rho)^2) \leq 2(h-1)\rho$. Applying the Lemma and a union bound, we have with probability at least $1-\frac{\delta}{2}$ that for all $i \in \calH$,
\begin{equation}\label{eq:hon-player-bits}
    |hon_i^{11} - \E[hon_i^{11}]| \leq 2\sqrt{(h-1)\rho \ln \tfrac{4n}{\delta}}.
\end{equation}

On the other hand, we have that $0 \leq mal_i^{11} \leq m$, so if we let $E_{mal,i}^{11} = \rho^2 m + (1-2\rho) d_i(\calM)$ (defined for convenience later), then $|mal_i^{11} - E_{mal,i}^{11}| \leq m$.

Applying the triangle inequality, the following holds over all $i \in \calH$:
\begin{align*}
&|hon_i^{11} - \E[hon_i^{11}] + mal_i^{11} - E_{mal,i}^{11}|
\leq m + 2\sqrt{ \rho n \ln \tfrac{4n}{\delta}} \\
& \implies |count_i^{11} - \rho^2 (n-1) - (1-2\rho) d_i| \leq m + 2\sqrt{ \rho n \ln \tfrac{4n}{\delta}} \\
& \implies |\hat{d}_i - d_i| \leq \frac{m + 2\sqrt{ \rho n \ln \tfrac{4n}{\delta}}}{1-2\rho}
\end{align*}
This proves the claim.
\end{proof}
Next, we show that honest users are not likely to be disqualified.
\begin{claim}\label{claim:honest-response-concentration-2} We have
\[
    \Pr[\forall \DO_i \in \calH.~|count_i^{01} - \rho(1-\rho)(n-1)| \geq \tau] \leq \frac{\delta}{2},
\]
\end{claim}
where $\tau = m + \sqrt{2 \rho n \ln \frac{4 n}{\delta}}$

\begin{proof}
Let $\DO_i$ be honest. Then, the quantity $hon_i^{01}$ consists of $h-1$ Bernoulli random variables drawn from $\rho(1-\rho)$. We have
\[
    \E[hon_i^{01}] = \rho(1-\rho)(h-1).
\]
 As defined in Lemma~\ref{lem:bern-concentration}, $v$ satisfies $\frac{1}{2}(h-1)\rho \leq P \leq (h-1)\rho$.
Applying the Lemma and a union bound, we have with probability $1-\frac{\delta}{2}$ that for all $i \in \calH$, 
\begin{equation}\label{eq:hon-player-bits-2}
    |hon_i^{01} - \E[hon_i^{01}]| \leq \sqrt{2\rho (h-1) \ln \tfrac{4n}{\delta}}
\end{equation}
Noticing that $|mal_i^{01} - m \rho(1-\rho)| \leq m$, we have by the triangle inequality that
\[
    |count_i^{01} - \rho(1-\rho)(n-1)| \geq m + \sqrt{2\rho n \ln \tfrac{4n}{\delta}}.
\]
This concludes the proof.
\end{proof}
Putting it together,
$m (\frac{e^\epsilon+1}{e^{\epsilon}-1}) + \sqrt{n}\frac{2 \sqrt{(e^\epsilon+1)\ln \frac{4n}{\delta}}}{e^\epsilon-1}$-honest error follows.

\noindent \textbf{Malicious Error.} 

When player $i$ is a malicious player, we can still prove a tight bound on $count_{i}^{11} + count_{i}^{01}$, and this combined with the check in \DegRRNaive{} means that his degree estimate will be accurate.

\begin{claim}\label{claim:mal-response-concentration}
We have
\begin{equation*}
    \Pr[\forall i \in \calM.~|count_i^{11} + count_i^{01} 
			-(1-2\rho)d_i - \rho(n-1)| \leq \tau ] \geq 1-\delta,
\end{equation*}
where $\tau = m + \sqrt{2 \rho n \ln \frac{4 n}{\delta}}$.
\end{claim}
\begin{proof}
Observe that $count_i^{11} + count_i^{01} = \sum_{j=1, j\neq i}^n q_{j}[i]$. Let $hon_i^{1}$ denote the sum of the $q_{j}[i]$ where $j$ is honest, and $mal_i^{1}$ denote the sum of the malicious players. By Fact~\ref{fact:rr-exp}, we have $\E[hon_i^1] = d_i(\calH) (1-2\rho) + h \rho$. Applying a union bound over Lemma~\ref{lem:bern-concentration}, for all $i \in \calM$, we have with probability at most $\delta$ that
\begin{equation}\label{eq:good-event-3}
    |hon_i^1 - \E[hon_i^1]| \geq \sqrt{2\rho n \ln \tfrac{2m}{\delta}}
\end{equation}
Because $|mal_i^1 - (1-2\rho)d_i(\calM) - \rho (m-1)| \leq m$, the claim follows from the triangle inequality.
\end{proof}

To conclude the proof, consider any malicious user $i \in \calM$ is not disqualified ($\hat{d}_i \neq \bottom$),
as if he is then the maximum malicious error event trivially happens. Thus, it must be true that $|count_i^{01} - (n-1)\rho(1-\rho)| \leq \tau$. However, given this and the event in Claim~\ref{claim:mal-response-concentration} holds, it follows by the triangle inequality that
\begin{align*}
    |count_i^{11}-(1-2\rho)d_i - \rho^2(n-1)| &\leq 2 \tau \\
    |\hat{d}_i - d_i| &\leq \frac{2 \tau }{1-2\rho}
\end{align*}
This establishes $2m (\frac{e^\epsilon+1}{e^{\epsilon}-1}) + 4\sqrt{n}\frac{ \sqrt{(e^\epsilon+1)\ln \frac{4n}{\delta}}}{e^\epsilon-1}$-malicious error.

\subsection{Proof of Theorem~\ref{thm:rrlapchecka3}}
\textbf{Honest Error.}\label{app:thm:rrlapchecka3}
By Claim~\ref{claim:honest-response-concentration-2}, the first check in \DegHybrid{} will not set $\hat{d}_i = \bottom$ for any honest user with probability at least $1-\frac{\delta}{4}$. 
The variables $\hat{d}_i^{rr}$ in $\DegHybrid$ behave identically to $\hat{d}_i$ in \DegRRCheck{}. By Claim~\ref{claim:honest-response-concentration-1} 
we have for all users, $|\hat{d}_i^{rr} - d_i| \leq \frac{m + 2 \sqrt{\rho n \ln \frac{8n}{\delta}}}{1-2\rho}$,
with probability at least $1-\frac{\delta}{4}$. 

By concentration of Laplace random variables, we have for all $i \in \calH$ that $|\hat{d}_i^{lap} - d_i| \leq \frac{1}{\epsilon} \ln \frac{2n}{\delta}$ with probability at least $1-\frac{\delta}{2}$, 
and by the triangle inequality we have $|\hat{d}_i^{lap} - \hat{d}_i^{rr}| \leq \frac{m + 2 \sqrt{\rho n \ln \frac{8n}{\delta}}}{1-2\rho} + \frac{1}{\epsilon}\ln \frac{2n}{\delta}$. Thus, the second check will not set $\hat{d}_i = \bottom$ assuming these events hold, and the estimator $\hat{d}_i$ satisfies the honest error bounds of $\hat{d}_i^{lap}$.

\textbf{Malicious Error.}
Following the same argument we saw in the  proof of malicious error of Theorem~\ref{thm:response:check}, we can have that, with probability at least $1-\frac{\delta}{2}$, for all malicious users $i \in \calM$, we have $|\tilde{d}_i^{rr} - d_i| \leq \frac{2\tau}{1-2\rho}$. Suppose that $\hat{d}_i$ is not set to be $\bottom$. This implies that $|\tilde{d}_i^{rr} - \tilde{d}_i^{lap}| \leq \frac{2\tau}{1-2\rho} + \frac{1}{\epsilon}\log \frac{2n}{\delta}$.  By the triangle inequality, this implies
\[
    |\tilde{d}_i^{rr} - d_i| \leq \frac{4\tau}{1-2\rho} + \frac{1}{\epsilon} \log \frac{2n}{\delta}.
\]
This establishes $\frac{4\tau}{1-2\rho} + \frac{1}{\epsilon} \log \frac{2n}{\delta}$ malicious error.

\subsection{Proof of Theorem~\ref{thm:input:laplace}}\label{app:thm:input:laplace}
\textbf{Honest Error.} The honest error guarantee follows in the same way as Theorem~\ref{thm:response:laplace}.
\\
\noindent \textbf{Malicious Error.} Consider a malicious user $\DO_i$, and let $m_i$ be the malicious degree estimate sent by $\DO_i$, with $0 \leq m_i \leq n-1$. The estimator is given by $\hd_i = m_i+\eta, \eta \sim Lap(\frac{1}{\epsilon})$. 
Thus, $\Pr[|d_i - m_i - \eta| \geq n-1] \leq \Pr[\eta > 0] \leq \frac{1}{2}$.

\subsection{Proof of Theorem~\ref{thm:b2a2_easy}} \label{app:thm:b2a2_easy}
\noindent \textbf{Honest error .}
We follow the honest error proof of Theorem~\ref{thm:response:naive}, with the following change. Observe that $mal_i$ consists of $|\calM_i|$ Bernoulli random variables of mean either $\rho$ or $1-\rho$. Thus, with probability $1-\frac{\delta}{2}$, we have $|mal_i - \E[mal_i]| \leq \sqrt{2m\ln \frac{4m}{\delta}}$ for all $i \in \calM$. 

Thus, we can show $|mal_i - E_{mal,i}| \leq (1-2\rho) |\calM_i|$, where $E_{mal,i} = \rho |\calM_i| + (1-2\rho)d_i(M_i)$.
Finishing the proof, we can show 
\[
|\hat{d}_i - d_i | \leq \frac{1}{1-2\rho} (\sqrt{2\rho n \ln \tfrac{4n}{\delta}} + \sqrt{2m \ln \tfrac{4m}{\delta}}) + m.
\]
\noindent \textbf{Malicious Error.}

In order for $|d_i - \hat{d}_i| = n-1$, it is necessary for $|count_i^{1} - \rho(n-1) - (1-2\rho)d_i| \geq (1-2\rho)(n-1)$. We have $count_i^1$ is a sum of $n-1$ Bernoulli random variables of mean either $\rho$ or $1-\rho$, so it can be written as $\mu + Z_i$, where $Z_i$ is approximately a normal random variable of mean $0$. Observe that, since $\mu$ and $\rho(n-1) + (1-2\rho)d_i$ are in the interval $[\rho (n-1), (1-\rho)(n-1)]$, it is impossible for the difference $\mu - \rho(n-1) + (1-2\rho)d_i$ to exceed $(1-2\rho)(n-1)$ unless $Z_i$ has the correct sign, which happens with probability at most $\frac{1}{2}$. This establishes $n-1$-malicious error with $\delta = \frac{1}{2}$.
\subsection{Proof of Theorem~\ref{thm:input:check}}\label{app:b3a2}

\textbf{Honest Error.}
Our proof follows that of Theorem~\ref{thm:response:check}. We are able to prove stronger versions of the claims.

\begin{claim}\label{claim:honest-input-concentration-1}
We have
\[
    \Pr[\forall i \in \calH.~|\hat{d}_i - d_i| \geq m+\frac{\sqrt{8\max\{\rho n, m\} \ln \frac{8n}{\delta}}}{1-2\rho}] \leq \frac{\delta}{2}.
\]
\end{claim}

\begin{proof}
We can control $hon_i^{11}$ in exactly the same way as in Claim~\ref{claim:honest-response-concentration-1}, so~\eqref{eq:hon-player-bits} holds with probability $1-\frac{\delta}{4}$, for all $i \in \calH$.
On the other hand, we know that $mal_i^{11}$ is now a sum of $d_i(\calM)$ Bernoulli random variables with bias either $(1-\rho)^2$ or $(1-\rho)\rho$, plus a sum of $m - d_i(\calM)$ Bernoulli random variables with bias either $\rho(1-\rho)$ or $\rho^2$. Thus, 
\begin{multline*}
	\rho(1-2\rho)d_i(\calM) + \rho^2 m \leq \E[mal_i^{11}] \\
	\leq (1-\rho)(1-2\rho)d_i(\calM) + \rho(1-\rho) m.
\end{multline*}
From this, we can show $|\E[mal_i^{11}] - E_{mal,i}^{11}| \leq (1-2\rho)m$, where $E_{mal,i}^{11} = \rho^2 m + (1-2\rho) d_i(\calM)$.
 Applying Hoeffding's inequality, we conclude that with probability at least $1 - \frac{\delta}{4}$, for all $i \in \calH$,
\[
    |mal_i^{11} - \E[mal_i^{11}]| \geq \sqrt{2m \ln \tfrac{8n}{\delta}}
\]
Thus, $|mal_i^{11} - E_{mal,i}^{11}| \leq (1-2\rho)m + \sqrt{2m \ln \tfrac{8n}{\delta}}$. Applying the triangle inequality, we obtain
\begin{multline*}
\Pr[|hon_i^{11} + mal_i^{11} - \E[hon_i^{11}] - E_{mal,i}^{11}| \\ \geq \sqrt{2m \ln \tfrac{8n}{\delta}} + (1-2\rho)m + 2\sqrt{ \rho n \ln \tfrac{8n}{\delta}}] \leq \tfrac{\delta}{2}.
\end{multline*}

The result follows in the same way as in Claim~\ref{claim:honest-response-concentration-1}.
\end{proof}

\begin{claim}\label{claim:honest-input-concentration-2}
We have
\[
    \Pr[\forall i \in \calH.~|count_i^{01} - \rho(1-\rho)(n-1)| \geq \tau] \leq \tfrac{\delta}{2},
\]
where $\tau = m(1-2\rho) + \sqrt{8 \max\{\rho n, m\} \ln \frac{8n}{\delta}}$.
\end{claim}

\begin{proof}
We can follow the same line of reasoning as Claim~\ref{claim:honest-response-concentration-2} and conclude that~\eqref{eq:hon-player-bits-2} holds.
Similar to Claim~\ref{claim:honest-input-concentration-1}, we can show that $|mal_i^{01} - \rho(1-\rho) m| \leq (1-2\rho)m + \sqrt{2m \ln \frac{8n}{\delta}}$ with probability at least $\frac{\delta}{4}$, and applying the triangle inequality, we see
\begin{equation*}
    \Pr[|count_i^{01} - \rho(1-\rho)n| \geq m(1-2\rho) + \sqrt{2m \ln \tfrac{8n}{\delta}}  + \sqrt{2\rho n \ln \tfrac{8n}{\delta}}] \leq \tfrac{\delta}{2}.
\end{equation*}
\end{proof}

The $2m+\frac{4\sqrt{2 \max\{\rho n, m\} \ln \frac{8n}{\delta}}}{1-2\rho}$-honest error guarantee follows from the union bound over the two claims.

\textbf{Malicious Error.}
When player $i$ is a malicious player, he is still subject to the following claim:

\begin{claim}\label{claim:mal-input-concentration}
We have
\begin{equation*}
    \Pr[\forall i \in \calM.~|count_i^{11} + count_i^{01} -(1-2\rho)d_i - \rho(n-1)| \leq \tau|] \geq 1-\delta,
\end{equation*}
where $\tau = m(1-2\rho) + \sqrt{8 \max\{\rho n, m\} \ln \frac{8n}{\delta}}$.
\end{claim}
\begin{proof}
Observe that $count_i^{11} + count_i^{01} = \sum_{j=1,j\neq i}^n q_{j}[i] = hon_i^1 + mal_i^1$. With the same argument as in Claim~\ref{claim:mal-response-concentration}, we know that~\eqref{eq:good-event-3} holds. Similarly, each random variable in $mal_i^1$ comes from either $\bern(\rho)$ or $\bern(1-\rho)$, and thus with probability at least $1-\frac{\delta}{2}$, for all $i \in \calM$
\[
    |mal_i^1 - \E[mal_i^1]| \leq \sqrt{2m \ln \tfrac{4m}{\delta}}
\]
Since $\E[mal_i^1] \in [\rho m, (1-\rho) m]$, 
This implies that $|mal_i^1 - \rho m | \leq  (1-2\rho)m + \sqrt{2m \ln \frac{4m}{\delta}}$. Thus, the claim follows.
\end{proof}

Having established this claim, we can prove $2m + 4\sqrt{2 \max\{\rho n, m\} \ln \frac{8n}{\delta}}$ malicious error using an identical method as in the proof of malicious error for Theorem~\ref{thm:response:check}.

\subsection{Proof of Theorem~\ref{thm:rrlapchecka2}}\label{app:thm:rrlapchecka2}
\textbf{Honest Error.} As input manipulation attacks are a subset of response manipulation attacks, the same  guarantee for honest error as Theorem~\ref{thm:rrlapchecka3} holds.

\textbf{Malicious Error.}
The proof of this is similar to the malicious error proof of Theorem~\ref{thm:rrlapchecka3}, using previous results in Theorem~\ref{thm:input:check}.
\fi

\end{document}